\documentclass[11pt,letterpaper]{article}
\usepackage[utf8]{inputenc}
\usepackage{fullpage}
\usepackage{amsthm}
\usepackage{xcolor}
\usepackage{mathtools,amsmath,amssymb}
\usepackage{mathrsfs}
\usepackage{algorithm, algpseudocode}
\usepackage{caption}
\usepackage{subcaption}
\usepackage{authblk}
\usepackage[colorlinks=true, allcolors=blue]{hyperref}
\usepackage[capitalise,nameinlink]{cleveref}
\usepackage[style=alphabetic, backend=biber, minalphanames=3, maxalphanames=5, maxbibnames=99, sorting=nyt]{biblatex}

\usepackage{pifont,mdframed}
\usepackage{tikz}
\usetikzlibrary{quantikz2}%need an updated TeX Live version for that
\usetikzlibrary{calc,positioning,arrows.meta}%,backgrounds,fit,decorations.pathreplacing,quantikz}

\usepackage{multirow}

\usepackage{bm}
\usepackage{enumitem}
\setlist[itemize]{topsep=3pt}
\setlist[enumerate]{topsep=3pt}
\usepackage{mathtools}
\usepackage{complexity}
\newlang{\ApxSim}{ApxSim}
\newlang{\LHPlusE}{LHPlusE}
\newlang{\LH}{LH}
\newlang{\kSSH}{kSSH}
\usepackage{dsfont}
\usepackage{float}
\usepackage{thm-restate}
\usepackage{bbm}
\usepackage{amsthm}
\usepackage{thmtools,thm-restate}

\usepackage[colorlinks=true, allcolors=blue]{hyperref}
\usepackage[nameinlink,capitalise]{cleveref}
\hypersetup{
    citecolor={violet}
}

\theoremstyle{plain}
\newtheorem{theorem}{Theorem}[section]
\numberwithin{equation}{section}

\newtheorem{corollary}[theorem]{Corollary}

\newtheorem{lemma}[theorem]{Lemma}
\newtheorem{claim}[theorem]{Claim}
\newtheorem{fact}[theorem]{Fact}

\newtheorem{question}[theorem]{Question}

\newtheorem{definition}[theorem]{Definition}
\newtheorem{problem}[theorem]{Decision Problem}

\newtheorem{remark}[theorem]{Remark}
\theoremstyle{plain}

\newcommand{\wval}{\textnormal{WVAL}}

\newcommand{\wmem}{\textnormal{WMEM}}

\newcommand{\cyes}{\textnormal{C}_{\textnormal{YES}}}
\newcommand{\pyes}{\textnormal{P}_{\textnormal{YES}}}
\newcommand{\pno}{\textnormal{P}_{\textnormal{NO}}}

\newcommand{\nul}{\nu_{\text{low}}}
\newcommand{\nuh}{\nu_{\text{high}}}

\DeclarePairedDelimiterXPP{\bigo}[1]{O}{(}{)}{}{#1}
\DeclarePairedDelimiterXPP{\bigomega}[1]{\Omega}{(}{)}{}{#1}

\newcommand{\defeq}{\stackrel{\mathrm{\scriptscriptstyle def}}{=}}

\renewcommand{\poly}{\textnormal{poly}}
\newcommand{\NN}{\mathbb{N}}
\newcommand{\CC}{\mathbb{C}}

\newcommand{\sep}{\mathsf{SEP}}

\newcommand{\density}{\mathsf{Density}}
\newcommand{\semicheck}{\mathsf{MatchCheck}}
\newcommand{\cnot}{\text{CNOT}}

\usepackage[margin=1.in]{geometry}
\usepackage{physics, graphicx}
\usepackage[compact]{titlesec}
\titlespacing{\subsection}{0pt}{1.5ex}{0ex}
\titlespacing{\subsubsection}{0pt}{1.5ex}{0ex}
\titlespacing{\subsubsection}{0pt}{1ex}{0ex}
\titlespacing{\paragraph}{0pt}{1.5ex}{1ex}
\usepackage{tikz}
\usetikzlibrary{quantikz2}
\usetikzlibrary{calc,positioning,arrows.meta}%,backgrounds,fit,decorations.pathreplacing,quantikz}
\renewcommand{\cw}{\mathcal{W}}
\newcommand{\cd}{\mathcal{D}}

\newcommand{\ccc}{{\mathcal{W}_{\ket{\textnormal{comp}}}}}
\newcommand{\ccs}{{\mathcal{W}_{\ket{\textnormal{dup}}}}}
\newcommand{\cpp}{{\mathcal{W}_{\ket{\mathsf{IS}}}}}

\newcommand{\cm}{{\mathcal{W}_{\mathsf{IS}}}}
\newcommand{\cg}{\mathcal{W}_{\ket{\textnormal{multisep}}}}
\newcommand{\cmg}{\mathcal{W}_{{\textnormal{multisep}}}}

\newcommand{\wproper}{\mathcal{W}_{\textnormal{proper}}}
\newcommand{\QMAP}{\QMA_{\mathsf{IS}}}
\newcommand{\QMAPP}{\QMA_{\ket{\mathsf{IS}}}}

\newcommand{\QMACP}{\QMA_{\ket{\textnormal{comp}}}}
\newcommand{\QMAMS}{\QMA_{\ket{\textnormal{multisep}}}}
\newcommand{\QMAMMS}{\QMA_{{\textnormal{multisep}}}}

\newcommand{\purerestrict}{\textnormal{rank1}}

\usepackage{diagbox}

\title{Quantum Merlin-Arthur with an internally separable proof}
 \author[1]{Roozbeh Bassirian\thanks{\href{mailto:roozbeh@uchicago.edu}{roozbeh@uchicago.edu}}}
 \author[1]{Bill Fefferman\thanks{\href{mailto:wjf@uchicago.edu}{wjf@uchicago.edu}}}
 \author[2]{Itai Leigh\thanks{\href{mailto:itai.leigh@mail.huji.ac.il}{itai.leigh@mail.huji.ac.il}}}
 \author[1]{Kunal Marwaha\thanks{\href{mailto:kmarw@uchicago.edu}{kmarw@uchicago.edu}}}
 \author[3]{Pei Wu\thanks{\href{mailto:pei.wu@psu.edu}{pei.wu@psu.edu}}}
 \affil[1]{University of Chicago}
 \affil[2]{Tel Aviv University}
 \affil[3]{Penn State University}
\date{}
\begin{document}
\maketitle
\vspace{-3em}
\begin{abstract}
We find a modification to $\QMA$ where having one quantum proof is strictly less powerful than having two unentangled proofs, assuming $\EXP \ne \NEXP$.
This gives a new route to prove $\QMA(2) = \NEXP$ that overcomes the primary drawback of a recent approach~\cite{jeronimo2023power,bassirian2023qmaplus} (QIP 2024).
Our modification endows
each proof with a form of \emph{multipartite} unentanglement: after tracing out one register, a small number of qubits are separable from the rest of the state.
\end{abstract}

\section{Introduction}
Entanglement is a fundamental property of quantum mechanics. One approach to understand this phenomenon is through the lens of computational complexity.
For example, characterizing the power of $\QMA(2)$ is now one of the central questions in quantum complexity theory, and has been unresolved for two decades (e.g.~\cite{qma2_defn,power_of_unentanglement,harrow2013testing}).
We make progress on this question by finding a nearby model of computation where one quantum proof is strictly less powerful than two \emph{unentangled} quantum proofs, unless $\EXP = \NEXP$.

$\QMA$ --- a quantum variant of $\NP$ (see \cite{DBLP:journals/sigact/Gharibian23} for a complete review) --- is the class of decision problems that can be decided with access to a quantum proof.
Here, a quantum computer (``Arthur'') communicates with a dishonest but all-powerful machine (``Merlin'').
Merlin sends a quantum state (``proof'') that is most likely to convince Arthur of a statement's veracity.
If the statement is true, there is a proof that convinces Arthur to accept with high probability; if it is false, all proofs make Arthur reject with high probability.

In $\QMA(2)$, 
the quantum proof is guaranteed to be 
\emph{not} entangled across a fixed bipartition ---
equivalently, the  two parts are given by two \emph{unentangled} Merlins.
It is natural to wonder what additional power comes from this entanglement structure.
Although there is evidence that $\QMA(2)$ is more powerful than $\QMA$~\cite{blier2010quantum,chen2010short,per12,qma2_yirka}, we still
only know the trivial bounds $\QMA \subseteq \QMA(2) \subseteq \NEXP$.

\begin{quote}
    {\hfil \textit{What is the power of the complexity class $\QMA(2)$?}}
    \vspace{0.85em}
    % \\
\end{quote}

In a recent attempt to answer this question, \cite{jeronimo2023power} proposed a modification to $\QMA$ they named $\QMA^+$, where the proof is required to have non-negative amplitudes. They used this model to exacerbate the difference between $\QMA$ and $\QMA(2)$, showing that $\QMA^+(2)$ with a constant completeness-soundness gap equals $\NEXP$. This result is particularly interesting because $\QMA^+(k)$ with a large enough gap is equal to $\QMA(k)$; therefore, if $\QMA^+(2)$ admits gap amplification, then $\QMA(2)=\NEXP$.
However, more recently it was discovered~\cite{bassirian2023qmaplus}
that $\QMA^+$ at some gap equals $\NEXP$ as well. 
Thus, if there exists a procedure to gap amplify $\QMA^+(2)$, it must crucially use unentanglement.\footnote{See \cite{jeronimo2024dimension} for the latest progress in this direction.}

In this work, we propose a different strategy to highlight the difference between $\QMA$ and $\QMA(2)$ by imposing additional entanglement structure to the quantum proof. 
We match the main results of the previous attempt while overcoming the primary drawback, giving a new approach to prove $\QMA(2) = \NEXP$.
Furthermore, the complexity classes we introduce 
can be seen as natural generalizations
of $\QMA(2)$, as they emphasize the role of \emph{multipartite} entanglement in computational complexity. We believe they are of independent interest.

\subsection{Results}
Consider the complexity class $\QMAP$, a variant to $\QMA$ where the quantum proof must have a fixed, separable \emph{subsystem}: that is, after tracing out all qubits beyond the subsystem, a predetermined set of $O(1)$ qubits are \emph{not} entangled with the rest of the subsystem.
Unlike in $\QMA(2)$, this guarantee refers to \emph{\underline{i}nternal} \underline{s}eparability: the two parts of the subsystem can be entangled, but only if the entanglement disappears when the subsystem is isolated.
See \Cref{fig:sep_vs_internal_sep} for a visual comparison.
We restrict one part of the subsystem to $O(1)$ qubits to zero in on the effect of \emph{internal} entanglement structure;
without the restriction, this class trivially contains $\QMA(2)$. 
\begin{figure}[t]
    \centering
    \includegraphics[width=0.8\textwidth]{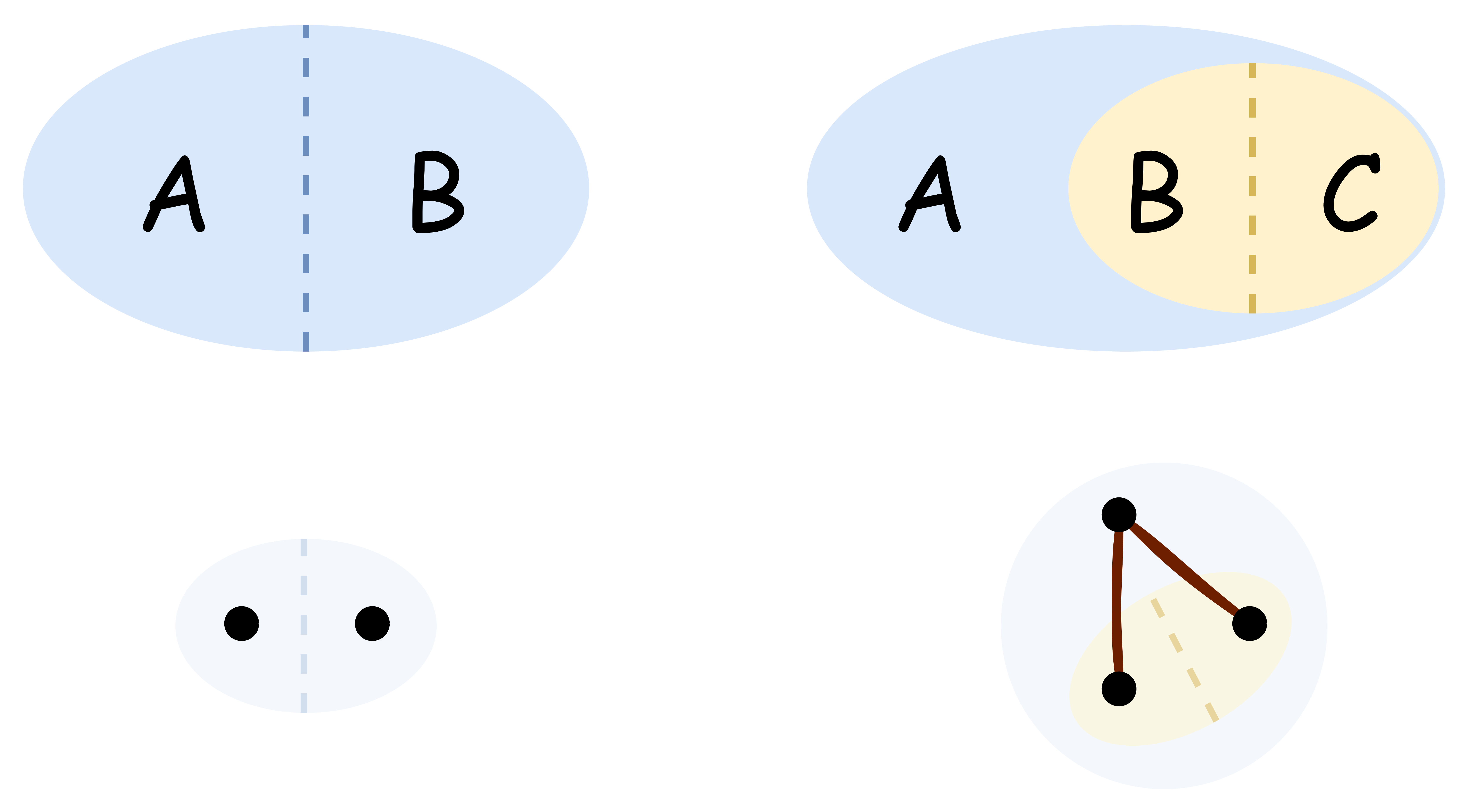}
    \caption{\footnotesize Cartoons of a separable state and a state with a separable subsystem. The left side depicts a separable state, where parts $A$ and $B$ share no entanglement. The right side depicts a state where \emph{after} tracing out part $A$, parts $B$ and $C$ share no entanglement; this state is in general entangled across every bipartition. The lower cartoons each display a graph where every vertex is a part ($A$, $B$, or $C$), and an edge between two parts $X,Y$ allows bipartite entanglement in the reduced state $\rho_{XY}$. On the right side, part A is entangled with both part B and part C. 
    }
    \label{fig:sep_vs_internal_sep}
\end{figure}

We first show that $\QMAP$, even with inverse exponentially small gap, is contained in $\EXP$:
\begin{theorem}[Upper bound]
\label{thm:upperbound}
    For any polynomial $p$ and $1 > c > s > 0$ such that $c - s > \frac{1}{2^{p(n)}}$, we have $\QMAP^{c,s} \subseteq \EXP$, where $c,s$ are the completeness and soundness parameters, respectively.\footnotemark
\end{theorem}
\footnotetext{We formally define $\QMAP$ in~\cref{sec:setup}. In a $\QMAP^{c,s}$ protocol, YES instances are accepted with probability at least $c$, and NO instances are accepted with probability at most $s$.}
We also define $\QMAP(2)$, where the verifier receives \emph{two} unentangled proofs, each with the above guarantee. 
In contrast with \Cref{thm:upperbound}, we find a constant-sized gap where $\QMAP(2)$ equals $\NEXP$.
As a consequence, $\QMAP$ and $\QMAP(2)$ cannot be equal unless $\EXP = \NEXP$.
\begin{theorem}[Lower bound]
\label{thm:lowerbound}
    There exist constants $1 > c > s > 0$ where $\QMAP^{c,s}(2) = \NEXP$.
\end{theorem}
\begin{corollary}
\label{cor:exp_not_nexp}
    If $\EXP \ne \NEXP$, then $\QMAP \ne \QMAP(2)$.
    \vspace{-0.7em}
\end{corollary}
We summarize these results in \Cref{tab:categorization}.
\Cref{thm:lowerbound} also implies that a quantum verifier with a specific guarantee on entanglement structure is \emph{exponentially} more powerful than a classical verifier.

\Cref{cor:exp_not_nexp} provides a new route to characterize the power of $\QMA(2)$.
For any constant $k$, $\QMAP(k)$ at large enough completeness-soundness gap is equal to $\QMA(k)$. 
This suggests that the guarantee of internal separability is not too strong.
\begin{claim}
\label{claim:almost_qmak}
    For any constant $k$, there exist constants $1 > c > s > 0$ so that $\QMAP^{c,s}(k) = \QMA(k)$. 
\end{claim}
By \Cref{claim:almost_qmak} and \Cref{thm:lowerbound}, gap amplification of $\QMAP(2)$ would imply $\QMA(2) = \NEXP$:
\begin{corollary}
\label{cor:gapamp_means_qma2_nexp}
    If $\QMAP(2)$ exhibits gap amplification, then $\QMA(2) = \NEXP$.
\end{corollary}
Although the same is true of $\QMA^+(2)$~\cite{jeronimo2023power}, we now know that the class $\QMA^+$ at some completeness-soundness gap is also equal to $\NEXP$, while it is equal to $\QMA$ at a larger gap. 
As a consequence, $\QMA^+$ \emph{cannot} have gap amplification unless $\QMA = \NEXP$~\cite{bassirian2023qmaplus}. 
This limits the techniques available to prove gap amplification of
$\QMA^+(2)$ in order to show $\QMA(2) = \NEXP$. 

By contrast, gap amplification of $\QMAP(2)$ can be plausibly found by studying gap amplification of $\QMAP$, since the latter implies no such complexity collapse.
This raises the exciting possibility of using \Cref{cor:gapamp_means_qma2_nexp} to fully characterize $\QMA(2)$.
\setlength{\tabcolsep}{0.5em} % for the horizontal padding
\renewcommand{\arraystretch}{1.2}% for the vertical padding
\begin{table}[h]
    \centering
    \begin{tabular}{c|c|c|}
\diagbox[trim=rl,font=\footnotesize\itshape,height=2\line]{internal structure}{external structure} 
&  \textsc{general} & \textsc{separable bipartition} \\
    \hline
        \textsc{general} & $\QMA$  & $\QMA(2)$\\
        \hline
    \textsc{separable subsystem} & $\QMAP \subseteq \EXP$ & $\QMAP(2) = \NEXP$ \\
    \hline
    \end{tabular}
        \caption{\footnotesize Categorizing quantum verification protocols by the guaranteed entanglement structures in the proof.}
    \label{tab:categorization}
\end{table}

\subsection{Other implications}
\paragraph{A complexity-theoretic framework to study multipartite entanglement.}
Despite much attention given to $\QMA(2)$, entanglement can be much richer than simply across a bipartition.
For example, the W state $\frac{1}{\sqrt{3}}\left( \ket{001} + \ket{010} + \ket{100} \right)$ and GHZ state $\frac{1}{\sqrt{2}} \left( \ket{000} + \ket{111} \right)$ are entangled in \emph{inequivalent} ways:
 both states are entangled across every bipartition, yet the states are not equivalent under LOCC~\cite{PhysRevA.62.062314}.
The same paper observes that 
tracing out a qubit from the W state gives an entangled state, but the same action on the GHZ state destroys all entanglement.

$\QMAP$ allows us to probe one type of multipartite (un)entanglement from the lens of computational complexity, defined as \emph{partial semi-separability} in \cite{Brassard_2001}.
It is intuitive to study the power of $\QMA$ with other kinds of multipartite entanglement guarantees. As a first step, we show that restricting to quantum proofs that are pure and \emph{multiseparable}\footnote{See \Cref{sec:multisep} for a definition of multiseparability.}~\cite{Thapliyal_1999}
can also decide $\NEXP$:
\begin{corollary}
\label{cor:multisep}
    There exist constants $1 > c > s > 0$ where $\QMA^{c,s}$ restricted to pure, multiseparable quantum proofs can decide $\NEXP$.
\end{corollary}
These results motivate a general question about the power of multipartite entanglement:
\begin{question}
\label{question:multipartite_unentanglement}
    Which notions of \emph{multipartite unentanglement} are computationally strong? Why?
\end{question}

\paragraph{A unifying perspective: $\QMA(2)$ and purity testing.}
Many works~\cite{lcv06,harrow2013testing,bks17,Yu_2022} use the fact that a protocol over separable states is equivalently a protocol over states with a purity constraint.
This suggests that \emph{purity testing} is at the heart of $\QMA(2)$'s power, even if it seems like a weaker primitive.
In our case, $\QMAP(2)$ is no more powerful than $\QMAP$ with a purity test.
The proof of \Cref{thm:lowerbound} only uses the two unentangled quantum proofs to implement a SWAP test, ensuring that Merlin sends a \emph{pure} quantum state.  
We show that optimizing over \emph{purifications} of separable states is enough to decide $\NEXP$, and even hard to approximate:
\begin{corollary}
    Given any $M\in \mathbb{R}^{n\times n}$, consider the optimization problem
    \begin{align*}
    \max_{\rho \in \cpp}\, \Tr[M \rho]\,,
    \end{align*}
    where $\cpp$ is the set of $n$-dimensional \emph{pure} states that the prover in $\QMAP$ can send. 
    Then there exists a constant $0<\alpha<1$ for which it is
    $\NP$-hard to compute an $\alpha$-approximation of this problem.
\end{corollary}

By contrast, optimizing over separable states (of $\poly(n)$ dimension) is known to be $\NP$-hard, but with gap decreasing with system size~\cite{gurvits2003classicaldeterministiccomplexityedmonds,power_of_unentanglement,gharibian2009strong,blier2010quantum,per12,GallNN12}.
If the gap in these problems could be made constant, then $\QMA(2) = \NEXP$.\footnote{One must also assume that the problems can be made succinct.}

More generally, with mild conditions on a set of quantum states $\cw$, we show 
how a $\QMA(2)$ protocol with each proof restricted to $\cw$ can implement a $\QMA$ protocol with a proof restricted to \emph{pure states} in $\cw$. 
We use this as a recipe
to generate other computational models where one quantum proof is provably weaker than two \emph{unentangled} quantum proofs, up to standard complexity-theoretic assumptions. 
We give two examples in \Cref{sec:more_unentanglement_power}, from the pure quantum proofs in $\QMA^+$ and
the \emph{proper states} of~\cite{power_of_unentanglement,beigi2009np}.
Each set represents the
pure states of some larger set $\cw$. 
Then, $\QMA_\cw$ can be decided with a semidefinite program, but $\QMA_\cw(2)$ can restrict to the more capable pure states.
This technique can likely generate other examples that demonstrate the power of unentanglement.

\subsection{Techniques}
To show \Cref{thm:upperbound}, we rely on a well-known duality of membership testing of a convex set and optimizing a convex function over this set. We use this duality to reduce any $\QMAP$ problem to membership testing of the set of internally separable quantum states $\cm$.
In order to use tools from convex optimization, we represent quantum states as vectors in a generalized Bloch sphere. We use a non-ellipsoidal reduction that permits an error analysis matching our setup~\cite{liu2007thesis,gharibian2009strong}. 

We then show how to test membership of $\cm$ in $\EXP$.
It is essential that only $O(1)$ qubits are unentangled from the rest of the subsystem.
Our algorithm tests membership of two convex sets whose intersection forms $\cm$; note that this requires testers with inverse exponential precision.
The first set represents all Bloch vectors of quantum states, which can be tested using exponentially many explicit conditions~(e.g. \cite{Kimura_2003}).
The second set represents all vectors whose reduced density matrix on the subsystem is separable.
This can be tested by solving the Doherty-Parrilo-Spedalieri SDP~\cite{dps_02,Doherty_2004,Navascu_s_2009} which decides if a state has a \emph{symmetric extension}. A careful error analysis shows that separability of a bipartite state $\rho$ with local dimensions $M,N$ can be
tested 
to error $\delta$ in time $\tilde{O}((2/\delta)^{9N} \poly(M,N))$~\cite{ioannou2007computationalcomplexityquantumseparability}. When $N$ is a fixed constant, this SDP runs in exponential time.

To prove \Cref{thm:lowerbound}, we consider $\QMAPP$, a variant of $\QMAP$ where the proof must also be a \emph{pure state}. Note that unlike $\QMA$ and $\QMA(2)$, restricting $\QMAP$ to pure states can change the power of the complexity class, as an internally separable ensemble of pure states may not decompose into an \emph{ensemble} of internally separable pure states.\footnote{For example, the same is true in $\QMA^+$: this is the difference between doubly non-negative matrices and completely positive matrices.} 
We then develop a $\QMAPP$ protocol for an $\NEXP$-complete problem, building on 
the proof of $\QMA^+ = \NEXP$ \cite{blier2010quantum,jeronimo2023power,bassirian2023qmaplus}.

As in \cite{jeronimo2023power,bassirian2023qmaplus}, we consider a succinct constraint satisfaction problem (CSP) with constant promise gap, which is $\NEXP$-hard by the PCP theorem~\cite{AS92,ALMSS98,harsha2004robust}. In the $\QMA^+$ protocol, one enforces a \emph{rigidity} property of the quantum proof: the state must be close to $\frac{1}{\sqrt{R}}\sum_{j \in [R]} \ket{j}\ket{a_j}$.
It is easy to ensure that all constraints $j \in [R]$ are included in the state (``density''), but the positive-amplitude guarantee is used to ensure there is at most one assignment $a_j$ per constraint (``quasirigidity'').
We find a way to instead enforce quasirigidity using both internal separability \emph{and} purity.
\begin{lemma}[informal]
    There is a test $A$ and a fixed quantum circuit $C$ such that for any internally separable $\ket{\psi}$, if $A(\ket{\psi}) \ge 1 - \epsilon$, then $C\ket{\psi}$ has $1 - O(\epsilon)$ squared overlap with a \emph{quasirigid} state.
\end{lemma}

In our $\QMAPP$ protocol, the verifier enforces tripartite proofs of the form $\sum_{j, a} z_{ja} \ket{j}\ket{a,j}\ket{a}$. 
This can be done by measuring all registers and accepting iff we see two copies of the same constraint $j$, and two copies of the same assignment $a$.
We then observe that having multiple assignments per constraint adds entanglement to the second and third registers, even after tracing out the first register. Thus, if the state is \emph{internally separable}, it has high overlap with a state with one assignment $a$ per constraint $j$. 
Notably, this test is inspired by recasting the $\QMA^+$ protocol in the language of ``superposition detection'' and non-collapsing measurement~\cite{bassirian2024superposition,aaronson_pdqma}, 
an unphysical 
operation related to a problem of Aaronson~\cite{scott_blogpost_pdqp}.

We then prove that $\QMAPP \subseteq \QMAP(2)$.
Without loss of generality, the best proof Merlin can send  in a $\QMAP(2)$ protocol is a product state $\rho_1 \otimes \rho_2$, since separable states form the convex hull of the set of product states.
The SWAP test accepts a product state with probability $\frac{1}{2}(1 + \Tr[\rho_1 \rho_2])$: this is $1$ only if $\rho_1, \rho_2$ are \emph{pure}.
The $\QMAP(2)$ simulator thus applies a SWAP test with some probability, and otherwise applies the $\QMAPP$ circuit.

\subsection{Related work}

\paragraph{The complexity class $\QMA(2)$}
One sign that $\QMA(2)$ is more powerful than $\QMA$ is that  $\QMA$ with a $O(\log n)$-qubit proof is equal to $\BQP$~\cite{marriott2005quantumarthurmerlingames}.
By contrast, $\QMA(2)$ with two unentangled $O(\log n)$-qubit proofs can decide $\NP$, although with inverse polynomially small completeness-soundness gap~\cite{power_of_unentanglement,beigi2009np,blier2010quantum,chen2010short,chiesa2011improved,GallNN12,per12}. Improving the completeness-soundness gap to a constant would likely imply $\QMA(2) = \NEXP$; see~\cite{jeronimo2023power,bassirian2023qmaplus,jeronimo2024dimension} for a recent attempt using quantum proofs with non-negative amplitudes ($\QMA^+$).
In general, one only needs two unentangled proofs since $\QMA(2) = \QMA(k)$~\cite{harrow2013testing}; the proof of this result uses an iterated SWAP test.
There are few complete problems known for $\QMA(2)$~\cite{chailloux,Hayden_2013,Gutoski_2015}; one proposal is the \emph{pure state marginal problem}, which is known to be $\QMA$-hard and in $\QMA(2)$~\cite{liu06cldm,lcv06,liu2007thesis,boson_qmacomplete,Yu_2021,Yu_2022}.

\paragraph{Convex optimization and quantum entanglement}
Some quantum papers build a Turing reduction using the duality between convex optimization and membership testing of a convex set, for example $\NP$-hardness of separability testing~\cite{gurvits2003classicaldeterministiccomplexityedmonds,ioannou2007computationalcomplexityquantumseparability,gharibian2009strong}, and $\QMA$-hardness of the mixed state marginal problem~\cite{liu06cldm,lcv06,liu2007thesis,boson_qmacomplete}.
This reduction can be implemented in several ways, including the ellipsoid method~\cite{yudin_nem_ellipsoid,grotschel2012geometric}, random walks~\cite{bertsimas_vempala}, and a perceptron-like algorithm~\cite{liu2007thesis}.
One can use a hierarchy of semidefinite programs to detect entanglement~\cite{dps_02,Doherty_2004,Navascu_s_2009,Harrow_2017_anand} or optimize over separable states~\cite{bcy11_definetti,bh12,bks17}. In some cases, carefully chosen $\epsilon$-nets provide similar performance~\cite{shi2011epsilonnetmethodoptimizationsseparable,barak2013roundingsumofsquaresrelaxations,brandao2015estimatingoperatornormsusing}.

\paragraph{Measuring quantum entanglement}
There are numerous ways to quantify entanglement; see the surveys \cite{zyczkowski2006introductionquantumentanglementgeometric,G_hne_2009,horodecki2009entanglement,horodecki2024multipartiteentanglement}. One approach considers the set of reachable states under transformations comprising only of local operations and classical communication (LOCC)~\cite{PhysRevA.62.062314,bennett2000exactAndMultipartite,vrana2015entanglementTransformations}.
Other works classify multipartite entanglement by whether a \emph{subsystem} is separable~\cite{Thapliyal_1999,Brassard_2001}.
Permutation-symmetric approaches include the Positive Partial Transpose (PPT) criterion~\cite{Peres_1996} and quantum de Finetti theorems~\cite{Caves_2002,Christandl_2007}. Mysteriously, both approaches identify entangled states that ``look'' separable under LOCC transformations~\cite{toth_gune_prl_permutationsymmetry,bcy11_definetti,bh12,Lancien_2013,Li_2015}.
There exist ensembles of low-entanglement quantum states that are computationally indistinguishable from maximally entangled quantum states~\cite{aaronson2023quantumpseudoentanglement}.

\subsection{Outline of the paper}
We describe our setup in \Cref{sec:setup}. We prove \Cref{thm:upperbound} in \Cref{sec:qmamix-in-exp}. We prove \Cref{thm:lowerbound} in \Cref{sec:qmaw_equals_nexp,sec:rank1restrict}; some details are deferred to \Cref{sec:probabilities,sec:padding}. We suggest some future directions in \Cref{sec:outlook}. 
We investigate the relationship of $\QMA(2)$ and purity testing in \Cref{sec:qma2_purity}.
We verify \Cref{claim:almost_qmak} in \Cref{sec:gapamplification}.
We describe a complete problem for any restriction of $\QMA$ to a set of quantum proofs $\cw$ in \Cref{sec:completeproblem}.
We find more computational models demonstrating the power of unentanglement in \Cref{sec:more_unentanglement_power}.
We prove \Cref{cor:multisep} in \Cref{sec:multisep}.

\section{Our setup}
\label{sec:setup}
Let $\cd(X)$ be the set of density matrices over register $X$. Let $\sep(X,Y)$ be the set of separable density matrices over registers $X,Y$. Let $Sym(X_1, \dots,  X_k)$ be the symmetric subspace over registers $X_1, \dots,  X_k$. Let $XY|Z$ refer to the bipartition with registers $X,Y$ in one part and $Z$ in the other.

\paragraph{$\QMA$ with a restricted set of proofs.}
Fix any set of quantum states $\cw$. 
Let $\QMA_\cw(m)$ as $\QMA(m)$ where Merlin \emph{must} send $m$ unentangled proofs from $\cw$.
Here we modify both completeness and soundness. Only modifying completeness sometimes has no effect, as seen in \cite{sqma}.
\begin{definition}[$\QMA_{\cw}(m)$]
\label{defn:qmaw}
Let $c,s: \mathbb{N} \to \mathbb{R}^+$ be polynomial-time computable functions, and $m,p$ be polynomials.
Fix an arbitrary set of quantum states $\cw$.

A promise problem $A = (A_{yes}, A_{no})$ is in $\QMA_{\cw}^{c,s}(m)$ if and only if there exists a polynomial $q$ and a polynomial-time uniform family of quantum circuits $\{Q_n\}$, where $Q_n$ takes as input a string $x \in \{0,1\}^n$,
$m$ $p(n)$-qubit quantum states in $\cw$, and $q(n)$ scratch qubits in state $\ket{0}^{\otimes q(n)}$, such that:
\begin{itemize}
    \item (completeness) If $x \in A_{yes}$, there exists a set of $m$ $p(n)$-qubit quantum states $\rho_1, \dots, \rho_m \in \cw$ such that $Q_n$ accepts $\ketbra{x}\otimes \left(\rho_1\otimes\cdots\otimes\rho_m\right)\otimes\ketbra{0}^{\otimes q(n)}$
    with probability at least $c(n)$.
    \item (soundness) If $x \in A_{no}$, then for any set of $m$ $p(n)$-qubit quantum states $\rho_1, \dots, \rho_m \in \cw$, $Q_n$ accepts $\ketbra{x}\otimes \left(\rho_1\otimes\cdots\otimes\rho_m\right)\otimes\ketbra{0}^{\otimes q(n)}$ with probability at most $s(n)$.
\end{itemize}
\end{definition}
\begin{remark}
    Following the convention in~\cite{jeronimo2023power}, we drop the $c,s$ when the completeness-soundness gap is allowed to be any absolute constant: $\QMA_\cw(m) \defeq \bigcup_{c-s = \Omega(1)} \QMA_\cw^{c,s}(m)$. However, note that $\QMA_\cw^{c,s}(m)$ is not in general equal to $\QMA_\cw^{c',s'}(m)$~\cite{bassirian2023qmaplus}.
\end{remark}

This work studies tripartite proofs. Label the three registers $A,B,C$. Consider two families of quantum states, where $k \in \mathbb{N}$ is the dimension of register $C$:
\begin{itemize}
    \item $\cm^k = \{\rho \ |\ \Tr_A[\rho] \in \sep(B,C)\}$: all mixed states that are \emph{internally separable} on $B,C$.
    \item $\cpp^k  = \{ \ket{\psi} \ | \ \Tr_A[\ketbra{\psi}] \in \sep(B,C) \}$: all \emph{purifications} of separable states on $B,C$.
\end{itemize}
We denote $\cm = \bigcup_{k \in \mathbb{N}} \cm^k$ and $\cpp = \bigcup_{k \in \mathbb{N}} \cpp^k$ to allow the dimension of register $C$ to be any fixed number.
This generates the following complexity classes:\footnote{
In the definition for $\QMAP$ and $\QMAPP$, the underlying Hilbert spaces of registers $A,B,C$ are omitted. To be more accurate, we could instead define $\QMAP^{a,b,k},\QMAPP^{a,b,k}$, where $a,b$ correspond to the number of qubits in registers $A$ and $B$ (some polynomial depending on $n$), and $k\in\NN$ is some constant representing the number of qubits for register $C$. 
However, to ease the notation, we put this information in $\cm$ and $\cpp$. 
}
\begin{align*}
    \QMAP(m) &\defeq \QMA_{\cm}(m), \\
    \QMAPP &\defeq \QMA_{\cpp}.
\end{align*}
We will only consider $\QMAPP$ with one proof, as we prove $\QMAPP = \NEXP$ in \Cref{sec:qmaw_equals_nexp}. This implies $\QMAPP(m) = \NEXP$ for all $m \ge 1$.

We occasionally select another property $P$ to choose a set of quantum states which we denote $\cw_P$. When this occurs, we use the shorthand $\QMA_P(m) \defeq \QMA_{\cw_P}(m)$. 
We use the label $\ket{P}$ to signify that the set $\cw_{\ket{P}}$ is pure.

\paragraph{Convex optimization.} In the area of convex optimization, a convex set $K$ is typically defined as a subset of $\mathbb{R}^d$ that is equipped with the standard (Euclidean) inner product.\footnote{For a self-contained introduction to convex optimization, see the textbook \cite{grotschel2012geometric}.} Let $B_\ell(\vec{p})$ be the ball of radius $\ell$ around point $\vec{p} \in \mathbb{R}^d$.
Some algorithms use special properties of $K$:
\begin{itemize}
    \item \emph{compactness}: By Heine-Borel, it is equivalent for $K$ to be closed and bounded.
    \item \emph{well-bounded $\vec{p}$-centeredness}: There is a known radius $R$ such that $K \subseteq B_R(\vec{0})$. Moreover, there is a known point $\vec{p}$ and radius $r$ such that $B_r(\vec{p}) \subseteq K$.
\end{itemize}

To work with this formalism, we represent a quantum state in $\cd(\mathbb{C}^M)$ by its generalized Bloch vector in $\mathbb{R}^{M^2-1}$; see \cite{Kimura_2003} for a review.
For example, fix any set $\{\sigma_i\}_{i \in [M^2 - 1]}$ of traceless Hermitian generators of $SU(M)$ satisfying $\Tr[\sigma_i \sigma_j] = 2 \delta_{ij}$.
Then we may decompose $\rho$ as
\begin{align}
\label{eqn:rho_decomposition}
    \rho = \frac{\mathbb{I}}{M} + \frac{1}{2}\sum_{i=1}^{M^2 - 1} \vec{r}_i \sigma_i\,.
\end{align}
We call $\vec{r} \in \mathbb{R}^{M^2 - 1}$ the (generalized) \emph{Bloch vector} of $\rho$. Note that a Bloch vector uniquely specifies a quantum state. We can thus identify a set of quantum states $J \subseteq \cd(\mathbb{C}^M)$ with the set of associated Bloch vectors $K \subseteq \mathbb{R}^{M^2 -1}$.
We only consider problems where $M$ is a power of $2$. Thus, choose the generators to be the nontrivial $\log M$-qubit Pauli words, each multiplied by the normalizing factor $2^{(1-\log M)/2}$ to preserve the identity $\Tr[\sigma_i \sigma_j] = 2 \delta_{ij}$.

Let $S(K, \epsilon)$ be the set of points $x$ 
that are at most $\epsilon$-far from $K$; i.e. $B_\epsilon(x) \cap K$ is nonempty.
Let $S(K, -\epsilon)$ be the set of points $x \in K$ that are at least distance $\epsilon$ from the boundary; i.e. $B_\epsilon(x) \subseteq K$.
We always have $S(K,-\epsilon) \subseteq K \subseteq S(K,\epsilon)$.

We work with two of the five ``basic problems''~\cite{grotschel2012geometric} in convex optimization. The first problem we use decides the maximum value of  a linear function on a convex set.
\begin{definition}[{Weak Validity ($\wval_\epsilon$(K))~\cite{grotschel2012geometric,liu2007thesis,gharibian2009strong}}]
Let $K \subseteq \mathbb{R}^d$ be a convex, compact, and well-bounded
$\vec{p}$-centered set. Then, given $\vec{c} \in \mathbb{Q}^d$ with $\|\vec{c}\,\|_2 = 1$, and $\gamma, \epsilon \in \mathbb{Q}$ such that $\epsilon > 0$, decide the following:
\begin{itemize}
    \item     If there exists $\vec{y} \in S(K, -\epsilon)$  with $\vec{c} \cdot \vec{y} \ge \gamma + \epsilon$, then output YES.
    \item     If for all $\vec{x} \in S(K, \epsilon)$, we have $\vec{c} \cdot \vec{x} \le \gamma - \epsilon$, then output NO.\footnote{We use the naming convention of~\cite{grotschel2012geometric} where $\wval$ is a decision problem. As in~\cite{liu2007thesis,gharibian2009strong}, we formulate $\wval$ and $\wmem$ as promise problems instead of multi-output problems.
    } 
\end{itemize}
\end{definition}
The second problem we use decides if an arbitrary point is in the convex set:
\begin{definition}[{Weak Membership ($\wmem_\beta$(K))~\cite{grotschel2012geometric,liu2007thesis,gharibian2009strong}}]
Let $K \subseteq \mathbb{R}^d$ be a convex, compact, and well-bounded
$\vec{p}$-centered set. 
Then given $y \in \mathbb{Q}^d$ and parameter $\beta \in \mathbb{Q}$ such that $\beta > 0$, decide the following:
\begin{itemize}
    \item If $y \in S(K, -\beta)$, then output YES.
    \item If $y \notin S(K, \beta)$, then output NO.
\end{itemize}
\end{definition}

We conclude this section with a few facts about Bloch vector representations of quantum states. To start, the norm of the Bloch vector is always bounded by a constant:
\begin{claim}[Folklore]
\label{claim:blochvectors_bounded}
    Every Bloch vector has norm at most $R = \sqrt{2}$.
\end{claim}
\begin{proof}
Fix any quantum state $\rho$, and let $\vec{r}$ be the associated Bloch vector. Then 
\begin{align*}
1 \ge \Tr[\rho^2] = \frac{1}{M} + \frac{1}{2} \sum_{i=1}^{M^2-1} \vec{r}_i~^2 \ge \frac{1}{2} \|\vec{r}\,\|^2\,.\tag*{\qedhere}
\end{align*}
\end{proof}
Moreover, Bloch vectors with small enough norm are separable across every bipartition:
\begin{fact}[\cite{Gurvits_2002_separable_ball}]
\label{fact:separable_everywhere_small_ball}
    Fix any quantum state $\rho \in \cd(\mathbb{C}^M)$, and let  $\vec{r}$ be the associated Bloch vector. If $\|\vec{r}\,\|^2 \le \frac{1}{M^2}$, then $\rho$ is separable across every bipartition.
\end{fact}
\begin{corollary}
    The set of Bloch vectors representing separable states in $\cd(\mathbb{C}^M, \mathbb{C}^N)$ is well-bounded $\vec{0}$-centered with $r = \frac{1}{M^2N^2}$ and $R = \sqrt{2}$. 
\end{corollary}

Finally, the Euclidean distance of two vectors is upper-bounded by a constant times the trace distance of their associated matrices:
\begin{claim}
\label{claim:euclidean_bloch_atmost_trace_distance}
    Consider two vectors $\vec{r}, \vec{t} \in \mathbb{R}^{M^2 - 1}$, and let $\rho$ and $\tau$ be the associated $M \times M$ matrices according to~\cref{eqn:rho_decomposition}. Then $\|\vec{r} - \vec{t}\,\|_2 \le \sqrt{2} \|\rho - \tau\|_1$.
\end{claim}
\begin{proof}
    Note that $\|\rho- \tau\|_1 \ge \|\rho - \tau\|_2$ by monotonicity of Schatten norms. Moreover,
    \begin{align*}
        \|\rho - \tau\|_2 
        = \sqrt{  \frac{1}{4} \sum_{i,j = 1}^{M^2 - 1} (\vec{r}_i - \vec{t}_i)(\vec{r}_j - \vec{t}_j) \Tr[\sigma_i \sigma_j] }
        = \sqrt{ \frac{1}{2} \sum_{i=1}^{M^2 - 1} (\vec{r}_i - \vec{t}_i)^2 } 
        = \frac{1}{\sqrt{2}} \|\vec{r} - \vec{t}\|_2\,.\tag*{\qedhere} 
    \end{align*}
\end{proof}

\section{Optimization over internally separable states: $\QMAP \subseteq \EXP$} \label{sec:qmamix-in-exp}
In this section, we show how to decide any $\QMAP$ problem in exponential time. This has three steps:
\begin{align*}
    \QMAP \stackrel{(1)}{\le_m} \wval(\cm) \stackrel{(2)}{\le_T} \wmem(\cm) \stackrel{(3)}{\subseteq} \EXP\,.
\end{align*}
First, we build a many-one reduction\footnote{Recall that a \emph{Turing} reduction from problem $A$ to problem $B$ is an algorithm to solve problem $A$ using an oracle for problem $B$. A \emph{many-one} reduction is a special type of Turing reduction where the oracle can only be called once, and the oracle's output must be returned by the algorithm.}
from any problem in $\QMAP$ to 
the convex optimization problem $\wval(\cm)$. Then, we use a Turing reduction from $\wval(\cm)$ to the membership testing problem $\wmem(\cm)$. We finally show that $\wmem(\cm)$ can be decided in $\EXP$.\footnote{To be precise, step (1) runs in exponential time and creates an 
exponentially long output. 
Steps (2) and (3) each run in polynomial time with respect to input size, which is 
exponential in the original $\QMAP$ problem's input length.
}

For any problem in $\QMAP$, there is a $\QMAP$ protocol and polynomial $p$ such that for input size $n$, the quantum proof in $\cm$ has $p(n)$ qubits.
In this section, we identify $\cm$ with its subset on $p(n)$ qubits, i.e. the set of allowed quantum proofs in this protocol.
By \Cref{sec:setup}, we can interchange $\cm$ with its projection to the generalized Bloch sphere (a subset of $\mathbb{R}^{4^{p(n)} - 1}$) in exponential time.
Therefore, in this section, we may view $\cm$ as a subset of real Euclidean space. 

\subsection{Showing that $\cm$ is well-behaved}
We first verify that $\cm$ has all the properties we need in our reduction.
\begin{claim}
\label{claim:wpurif_convex_properties}
    $\cm$ is a convex, well-bounded $\vec{0}$-centered set with radii $R = \sqrt{2}, r = \frac{1}{4^{p(n)}}$.
\end{claim}
\begin{proof}
We first check convexity. Consider $\rho_1, \rho_2 \in \cm$. Since both $\Tr_A[\rho_1]$ and $\Tr_A[\rho_2]$ are separable, then so is $\Tr_A[p \rho_1 + (1-p) \rho_2]$. So the state $p \rho_1 + (1-p) \rho_2 \in \cm$. This establishes that $\cm$ is convex when viewed as density matrices. Since the transformation from density matrix to Bloch vector is linear, $\cm$ is also convex when viewed as Bloch vectors.

By \Cref{claim:blochvectors_bounded}, all Bloch vectors have norm at most $\sqrt{2}$. Moreover, by \Cref{fact:separable_everywhere_small_ball}, every Bloch vector in $\mathbb{R}^{M^2 - 1}$ with norm at most $\frac{1}{M^2}$ is separable across every bipartition. This implies that they are separable across the $AB|C$ bipartition, 
and are thus internally separable. So $\cm$ is well-bounded $\vec{0}$-centered with $r = \frac{1}{4^{p(n)}}$ and $R = \sqrt{2}$.
\end{proof}

It will be useful to represent $\cm$ as the intersection of two sets $\cm = K_1 \cap K_2$. Let $K_1$ be the set of all Bloch vectors corresponding to density matrices on $p(n)$ qubits; i.e. the projection of $\cd(\mathbb{C}^{2^{p(n)}})$ onto $\mathbb{R}^{4^{p(n)} - 1}$. 
Let $K_2$ be the set of all vectors in $\mathbb{R}^{4^{p(n)} - 1}$ whose associated vector on the subsystem represents a separable state (up to normalization), and where all other coordinates have magnitude at most $2$. This works for the following reason:
\begin{remark}
\label{remark:blochvector_subsystem}
When the Bloch vector encoding uses the generators $\{\sigma_i\}$ associated with Pauli words, the Bloch vector of a \emph{subsystem} is a subset of the coordinates of the original Bloch vector, multiplied to a global normalization.\footnote{The normalizing factor is exactly sqrt of the dimension of register $A$, to preserve the identity $\Tr[\sigma_i \sigma_j] = 2\delta_{ij}$.}
In particular, $\vec{r}_i = \Tr[\rho \sigma_i]$ corresponds to an element of the Bloch vector of the subsystem if and only if $\sigma_i$ is $\mathbb{I}$ on every qubit beyond the subsystem.
\end{remark}
\begin{claim}
\label{claim:k1k2properties}
    $K_1$ and $K_2$ are convex, compact, well-bounded $\vec{0}$-centered sets, with associated radii $R_1 = \sqrt{2}, R_2 = 4^{p(n)}, r = \frac{1}{4^{p(n)}}$.
\end{claim}
\begin{proof}
We start with convexity. The set of all $p(n)$-qubit states is convex; since the transformation from density matrix to Bloch vector is linear, $K_1$ is also convex. The set of separable states is convex, by the same argument, so are the associated Bloch vectors. $K_2$ is the set of all vectors which match a separable Bloch vector on a fixed set of coordinates, intersected with the solid hypercube $[-2,2]^{4^{p(n)}-1}$. Both sets are convex (and $K_2$ is non-empty), so $K_2$ is convex.

All Bloch vectors have norm at most $\sqrt{2}$. Vectors in $K_2$ have norm at most $\sqrt{2 \cdot(4^{p(n)}-1)}$. Since we have $\cm \subseteq K_1$ and $\cm \subseteq K_2$, both $K_1$ and $K_2$ contain a $\vec{0}$-centered ball of radius $r = \frac{1}{4^{p(n)}}$ by \Cref{claim:wpurif_convex_properties}.

The set of all density matrices is compact (e.g. \cite{zyczkowski2006introductionquantumentanglementgeometric}). 
Since the transformation from density matrix to Bloch vectors is linear, $K_1$ is also compact.
By Heine-Borel, $K_2$ is compact iff it is closed and bounded. It remains to show closure. The set of Bloch vectors representing separable states is closed (e.g. \cite{gharibian2009strong}). Note that $K_2$ is the Cartesian product of this set with the closed set $[-2,2]$ on each other coordinate. Since the product of closed sets is closed, $K_2$ is closed.
\end{proof}
\begin{claim}
\label{claim:wpurif_compact}
    $\cm$ is compact.
\end{claim}
\begin{proof}
By Heine-Borel, we can prove compactness by showing that $\cm$ is closed and bounded. $\cm$ is bounded because the norm is always at most $R$. $\cm$ is closed because it is the intersection of two closed sets $K_1,K_2$ (\Cref{claim:k1k2properties}).
\end{proof}

\subsection{Proving the reduction}
The first step is straightforward. Recall the definition of $\wval$ from \Cref{sec:setup}. Observe that the maximum acceptance probability of any $\QMAP$ protocol with verifier circuit $V$ can be written as $\max_{\rho \in \cm} \Tr[V \rho]$.
Thus, we can decide this problem if we can decide the value $\max_{\rho \in \cm} \Tr[V \rho]$ with good enough precision.

\begin{claim}
    Every problem in $\QMAP$ with completeness-soundness gap $\Delta$ has an exponential-time reduction to $\wval_{\epsilon}(\cm)$, where $\epsilon = \Omega(\frac{\Delta}{2^{\poly(n)}})$.
\end{claim}
\begin{proof}
Suppose the verification circuit is $V$, and let the proof be a state on $p(n)$ qubits.
We represent $V$ as a vector $\vec{c}$, so the acceptance probability of $V$ on $\rho$ is affine  to the inner product of $\vec{c}$ with the Bloch vector representation of $\rho$.
The quantity $\Tr[V \rho]$ can be decomposed via \cref{eqn:rho_decomposition} as
\begin{align*}
    \Tr[V \rho] = \frac{1}{2^{p(n)}}\Tr[V] + \frac{1}{2} \sum_{i=1}^{4^{p(n)}-1} \vec{r}_i \cdot \Tr[V \sigma_i]\,,
\end{align*}
where $\vec{r} \in \mathbb{R}^{4^{p(n)}-1}$ is the Bloch vector representation of $\rho$.
Let $\vec{d}$ be the vector where $\vec{d}_i = \Tr[V \sigma_i]$, and choose $\vec{c} \defeq \frac{\vec{d}}{\|\vec{d}\,\|}$. We now argue that there is a small enough $\epsilon$ so that deciding $\wval_\epsilon(\cm)$ with $\vec{c}$ decides the $\QMAP$ problem.

Since $V$ is generated by a polynomial-time Turing machine, there is a polynomial $q$ such that each matrix element of $V$ has norm at most $2^{q(n)}$. So, each vector element $\vec{d}_i = \Tr[V \sigma_i]$ has norm at most $2^{p(n) + q(n)}$. This implies that $\|\vec{d}\,\|$ is at most $2^{2p(n)+q(n)}$.

Suppose the $\QMAP$ problem has soundness $s$. Then in soundness, for any state $\rho \in \cm$, we have $\Tr[V \rho] \le s$, or equivalently,
\begin{align}
    \vec{c} \cdot \vec{r} = \frac{2}{\|\vec{d}\,\|} \left( \Tr[V \rho] - \frac{1}{2^{p(n)}} \Tr[V] \right) \le  \frac{2}{\|\vec{d}\,\|} \left( s - \frac{1}{2^{p(n)}} \Tr[V] \right)\,.
    \label{eq:wval-reduction-soundness}
\end{align}
The $\QMAP$ problem has completeness $s + \Delta$ by definition. In this case, there exists $\rho \in \cm$ where $\Tr[V \rho] \ge s + \Delta$, or equivalently, 
\begin{align}
     \vec{c} \cdot \vec{r} \ge \frac{2}{\|\vec{d}\,\|} \left( s - \frac{1}{2^{p(n)}} \Tr[V] \right) + \frac{2}{\|\vec{d}\,\|} \Delta \,.
     \label{eq:wval-reduction-completeness}
\end{align}

Define the parameter 
\[\gamma \defeq \frac{2}{\|\vec{d}\,\|} \left( s - \frac{1}{2^{p(n)}} \Tr[V] \right) + \frac{1}{\|\vec{d}\,\|}\Delta.\] 
We verify this problem can be solved by the \emph{weak} optimization problem $\wval_{\epsilon/4^{p(n)+1}}(\cm)$ for 
\begin{align*}
    \epsilon \defeq \min\left(\frac{1}{2 \cdot 4^{p(n)}}, \frac{\Delta}{2\|\vec{d}\,\|}\right)\,.
\end{align*}

\begin{itemize}
\item
In the soundness case, consider $\vec{y} \in S(\cm, \epsilon)$. By definition, there exists a vector $\vec{x} \in \cm$ such that $\|\vec{y} - \vec{x} \| \le \epsilon$.
So by Cauchy-Schwarz, for unit vector $\vec c$, we have
\begin{align*}
    \vec{c} \cdot \vec{y} 
    = \vec{c} \cdot \vec{x} + \vec{c} \cdot \left(\vec{y} - \vec{x} \right)
    \le \vec{c} \cdot \vec{x}  + \|\vec{c}\,\| \|\vec{y} - \vec{x} \| \le  \vec{c} \cdot \vec{x} + \epsilon\,.
\end{align*}
Since for any $\vec{x} \in \cm$, we have $\vec{c} \cdot \vec{x} \le \gamma - \frac{\Delta}{\|\vec{d}\,\|} \le \gamma - 2\epsilon$ by \cref{eq:wval-reduction-soundness}, the above inequality implies $\vec{c} \cdot \vec{y} \le \gamma - \epsilon$.

\item
In completeness, there is an $\vec{r} \in \cm$ such that $\vec{c} \cdot \vec{r} \ge \gamma + \frac{\Delta}{\|\vec{d}\,\|} \ge \gamma + 2\epsilon$ by \cref{eq:wval-reduction-completeness}. We split into two cases, depending on the norm of $\vec{r}$: 
    \begin{itemize}
        \item If $\|\vec{r}\,\| \le \epsilon \le \frac{1}{2 \cdot 4^{p(n)}}$, then it must be in $S(\cm, -\frac{1}{2 \cdot 4^{p(n)}}) \subseteq S(\cm, -\epsilon)$. 
        This is because by~\Cref{claim:wpurif_convex_properties}, all vectors of norm at most $\frac{1}{4^{p(n)}}$ are in $\cm$.
    \item Otherwise, define $\vec{z} \defeq \vec{r} \cdot \frac{\|\vec{r}\,\| - \epsilon}{\|\vec{r}\,\|}$. Since $\vec{0} \in \cm$, $\vec{z} \in \cm$ by convexity. 
    Moreover, 
    \begin{align*}
        \vec{c} \cdot \vec{z} = \vec c \cdot \vec r + \vec c \cdot (\vec z- \vec r)  \ge \vec c \cdot \vec r - \epsilon \ge \gamma +\epsilon\,.
    \end{align*}
    Finally, by a geometric argument\footnote{
    Consider a $\vec{0}$-centered convex set $K$ with inner radius $r$. Then for any vector $\vec{x} \in K$ and $0 \le \alpha \le 1$, the set $K$ contains a ball centered at $(1-\alpha) \vec{x}$ of radius at least $\frac{1}{2} \alpha r$.
    },
    $B_\ell(\vec{z}) \subseteq \cm$ for $\ell = \frac{1}{2} \cdot \frac{\epsilon}{\|\vec{r}\,\|} \cdot \frac{1}{4^{p(n)}} \ge \frac{\epsilon}{4 \cdot 4^{p(n)}}$.
    So then $\vec{z} \in S(\cm, -\frac{\epsilon}{4 \cdot 4^{p(n)}})$.
    \qedhere
    \end{itemize}
\end{itemize}
\end{proof}

The second step uses a duality between optimization over a convex set $K$ and membership testing of $K$. In our case, we need a reduction from optimization to membership testing. We use a reduction that runs in exponential time given our setup.
The runtime of this reduction depends on \emph{the encoding size} of the convex set $K$, denoted by $\langle K \rangle$. This is defined as the dimension plus the number of bits required to specify $r, R, \vec{p}$~\cite{gurvits2003classicaldeterministiccomplexityedmonds,gharibian2009strong}. In our case, $\langle \cm \rangle$ is at most $O(2^{\poly(n)})$.
\begin{theorem}[{\cite[Proposition 2.8]{liu2007thesis}}]
Let $K \subseteq \mathbb{R}^d$ be a convex, compact, and well-bounded $\vec{p}$-centered set with associated radii $(R, r)$.
Given an instance $\Pi = (K, c, \gamma, \epsilon)$ of $\wval_\epsilon(K)$ with $0 < \epsilon < 1$, there exists an algorithm which runs in time $\poly(\langle K \rangle, R, \lceil 1/\epsilon \rceil)$, and solves $\Pi$ using an oracle for $\wmem_{\beta}(K)$ with $1/\beta = O(\poly(d, \epsilon, R, 1/r))$.
\end{theorem}

Putting the first two steps together, we obtain the following:
\begin{corollary}
\label{cor:start_to_wmem}
    Every problem in $\QMAP$ with completeness-soundness gap $\Delta$ has an exponential-time reduction to $\wmem_{\beta}(\cm)$, where $\beta = \Omega(\frac{\Delta}{2^{\poly(n)}})$.
\end{corollary}

In the third step, we show that $\wmem(\cm)$ can be solved in exponential time. We do this by testing membership of convex sets $K_1,K_2$ whose intersection $K_1 \cap K_2 = \cm$. 
\begin{theorem}[{\cite[Section 4.7]{grotschel2012geometric}}]
\label{thm:solve_wmem_through_intersection}
For each $i \in \{1,2\}$, let $K_i \subseteq \mathbb{R}^d$ be a convex, compact, and well-bounded $\vec{p_i}$-centered set with associated radii $(R_i, r_i)$. Suppose we are given an $r_3 \in \mathbb{Q}$ such that $K_1 \cap K_2$ contains a ball with radius $r_3$. Then $K_1 \cap K_2$ is a well-bounded convex body, and given a instance $\Pi = (K_1 \cap K_2, y, \beta)$ of $\wmem_\beta(K)$, there exists an algorithm which runs in constant time and solves $\Pi$ using an oracle for $\wmem_{\zeta}(K_1)$ and $\wmem_{\zeta}(K_2)$, with  $\zeta = \beta \cdot \frac{r_3}{2 \min(R_1, R_2)}$.
\end{theorem}

To decide $\wmem(K_2)$, we must test that each coordinate of the vector has magnitude at most $2$, and that the vector's projection to the subsystem (as in \Cref{remark:blochvector_subsystem}) is proportional to a Bloch vector of a separable state. 
Fortunately, since one part of the subsystem has $O(1)$ qubits, we can test the latter using an algorithm that runs in time polynomial in dimension, i.e. exponential time.

There are many known algorithms to test whether a Bloch vector is separable; see \cite{ioannou2007computationalcomplexityquantumseparability} for a (slightly dated) comparison. In this work, we rely on an algorithm which searches over \emph{symmetric extensions} of a state. This algorithm takes an input $\rho \in \cd(X,Y)$ and tries to find a \emph{extension}; i.e. a state $\sigma \in \mathcal{D}(X,Y_1, \dots, Y_k)$ such that the reduced density matrix of $\sigma$ on $X$ and any $Y$ register is equal to $\rho$. A state is separable if and only if it is $k$-extendable for all $k$, so this test always succeeds at large enough $k$ (although this may not be polynomial in the number of qubits).
This algorithm was initiated by works of Doherty-Parrilo-Spedalieri \cite{dps_02,Doherty_2004} and the quantum de Finetti theorem~\cite{Caves_2002,Christandl_2007}, but see follow-ups  \cite{Navascu_s_2009,Brand_o_2012_search,Harrow_2017_anand} which find ways to quantify and reduce the value of $k$ required to test separability.
\begin{claim}[{\cite[Corollary 1 and Section 3.5]{ioannou2007computationalcomplexityquantumseparability}}]
\label{claim:entanglementtest}
    Let $\mathcal{S}_{N_1,N_2}$ be the set of separable states of dimension $N_1 \times N_2$. Then there is an algorithm to decide $\wmem_{\beta}(\mathcal{S}_{N_1,N_2})$ (with respect to the trace distance) in time $\poly((1/\beta)^{9N_2}, N_1, N_2)$.
\end{claim}
By \Cref{claim:euclidean_bloch_atmost_trace_distance}, we know that \Cref{claim:entanglementtest} also holds with respect to Euclidean distance of the real vectors. As a consequence:
\begin{theorem}
\label{cor:wmem_k2}
    We can decide $\wmem_\zeta(K_2)$ in exponential time, for any $\zeta$ that is at least an inverse exponential in $n$.
\end{theorem}
\begin{proof}
Let $M = 2^{p(n)}$ be such that $K_2 \subseteq \mathbb{R}^{M^2 - 1}$. Consider any vector $\vec{r} \in \mathbb{R}^{M^2 - 1}$. We first trivially check if any coordinate is more than $2$; if it is, we stop and report that $\vec{r} \notin K_2$. Otherwise, we consider the projection of $\vec{r}$ to the $BC$ subsystem, which is a vector on fewer dimensions. 
By \Cref{claim:entanglementtest}, we can check this reduced vector is separable up to error $\beta \defeq \frac{\zeta}{2M}$ in time $\poly((2M/\zeta)^{9\dim(C)}, \dim(B), \dim(C))$. Since $\dim(C)$ is at most some uniform constant, this runtime is $\poly(1/\zeta, M)$, i.e. exponential in $n$.

We remark on the error tolerance.
Let $\vert^{BC}$ represent the projection onto the $BC$ subsystem.
By \Cref{remark:blochvector_subsystem}, $K_2$ is the direct sum of $K_2\vert^{BC}$ and a solid hypercube (for the other coordinates).
\begin{itemize}
    \item In the YES case, the input $\vec{x} \in S(K_2, -\zeta)$. So $K_2\vert^{BC}$ contains a ball centered at $\vec{x}\vert^{BC}$ of radius $\zeta$: $\vec{x}\vert^{BC} \in S(K_2\vert^{BC}, -\zeta)$. This is correctly decided since \Cref{claim:entanglementtest} has error at most $\frac{\zeta}{M}$, and the normalization factor is at most $M$.
    \item In the NO case, the input $\vec{x} \notin S(K_2, \zeta)$. We can check each coordinate of $\vec{x}$ is at most $2$ up to bit precision, i.e. inverse double exponential in $n$. If this check succeeds, then $\vec{x}\vert^{BC}$ must be at least distance $\frac{\zeta}{2}$ from $K_2\vert^{BC}$: $\vec{x}\vert^{BC} \notin S(K_2\vert^{BC}, \frac{\zeta}{2})$.
    This is correctly decided since \Cref{claim:entanglementtest} has error at most $\frac{\zeta}{2M}$, and the normalization factor is at most $M$.\qedhere
\end{itemize}
\end{proof}

To decide $\wmem(K_1)$, we must test whether a vector is a true Bloch vector. 
This amounts to ensuring that the associated density matrix is positive semi-definite, as it is already trace-1 and Hermitian. There are several ways to test this property; for example, we can use \Cref{claim:entanglementtest} with $N_1 = M$ and $N_2 = 1$.
\begin{fact}
\label{fact:wmem_k1}
    We can decide $\wmem_\zeta(K_1)$ in exponential time, for any $\zeta$ that is at least an inverse exponential in $n$.
\end{fact}
\begin{remark}
    In a model of computation that supports exact arithmetic, we can analytically verify that a vector is a true Bloch vector~\cite{Kimura_2003}.
     Choose $\vec{r} \in \mathbb{R}^{M^2 - 1}$ and a set $\{\sigma_i\}_{i \in [M^2-1]}$ of traceless Hermitian generators of $SU(M)$ satisfying $\Tr[\sigma_i \sigma_j] = 2\delta_{ij}$. Let $P$ be the $M \times M$ matrix
    \begin{align*}
        P \defeq \frac{\mathbb{I}}{M} + \frac{1}{2} \sum_{i=1}^{M^2 - 1} \vec{r}_i\sigma_i\,.
    \end{align*}
    Then $P$ is positive semidefinite if and only if $a_k \ge 0$ for all $k \in [M]$, where $a_0 = 1$, and
    \begin{align*}
        k \cdot a_k = \sum_{q=1}^k (-1)^{q-1} \Tr[P^q] a_{k-q}\,.
    \end{align*}
\end{remark}

We can now state the main result of this section:
\begin{theorem}
\label{thm:exp_upper_bound}
    For all $0 \le c < s \le 1$ where $c - s = \Omega(\frac{1}{2^{\poly(n)}})$, we have $\QMAP^{c,s} \subseteq \EXP$.
\end{theorem}
\begin{proof}
Recall that $\cm = K_1 \cap K_2$, where $K_1$ and $K_2$ are convex, compact, well-bounded $\vec{0}$-centered sets with associated radii $R_1 = \sqrt{2}, R_2 = 4^{p(n)}, r = \frac{1}{4^{p(n)}}$ by \Cref{claim:k1k2properties}. 
    Both $\wmem_\zeta(K_1)$ and $\wmem_\zeta(K_2)$ can be solved in time exponential in $n$, for any $\zeta$ at least an inverse exponential in $n$, by \Cref{fact:wmem_k1} and \Cref{cor:wmem_k2}.
    Moreover, $\cm$ contains a $\vec{0}$-centered ball of radius $r_3 = \frac{1}{4^{p(n)}}$.
    Thus, by \Cref{thm:solve_wmem_through_intersection}, we can solve $\wmem_{\beta}(\cm)$ for any $\beta \ge \zeta \cdot 4^{p(n)+1} \ge \zeta \cdot \frac{2 \min(R_1, R_2)}{r_3}$; i.e. any $\beta$ at least inverse exponential in $n$. 
    
    Consider a problem in $\QMAP$ with completeness-soundness gap $\Delta$. By \Cref{cor:start_to_wmem}, there is an exponential-time reduction to $\wmem_{\beta}(\cm)$, where $\beta = \Omega(\frac{\Delta}{2^{\poly(n)}})$. 
    When $\Delta = \Omega(\frac{1}{2^{\poly(n)}})$, so is $\beta$, and this problem can be decided in exponential time.
\end{proof}

\begin{remark}
It is likely that \Cref{thm:exp_upper_bound} can be improved to a completeness-soundness gap that is inverse \emph{double} exponential in $n$.
This would require two changes. First, the second step needs a reduction that runs in $\poly(\log 1/\epsilon)$ time; this may be satisfied by the shallow-cut ellipsoid method~\cite{grotschel2012geometric,liu2007thesis}.
Second, membership testing of $K_2$ needs a program that runs in $\poly(\log 1/\epsilon)$ time; this can plausibly be done using an improvement to the Doherty-Parrilo-Spedalieri SDP~\cite{Harrow_2017_anand}.
\end{remark}

\section{Purity and internal separability: $\QMAPP = \NEXP$}
\label{sec:qmaw_equals_nexp}
In this section we prove there exist constants $c,s$ representing completeness and soundness for which $\QMAPP^{c,s}$ (the variant of $\QMAP^{c,s}$ restricted to pure quantum proofs) can decide $\NEXP$.

The starting point of our approach is the Blier-Tapp protocol to decide the $\NP$-hard  \textsc{3-COLORING} 
graph problem~\cite{blier2010quantum} in $\QMA_{\log}(2)$, i.e. $\QMA(2)$ with two unentangled quantum proofs each on $O(\log n)$ qubits.
In their protocol, the classical $\NP$ proof is encoded in $O(\log n)$-qubit quantum states of the form $\ket{\psi} = \sum_{v \in V} \ket{v}\ket{c_v}$, where $V$ is the set of vertices of the graph and $c_v$ is a color assigned to vertex $v$. So $\ket\psi$ encodes a coloring to the graph.
With two copies of a quantum proof in this form, the verifier can measure both in the computational basis, obtaining two random vertices $v,w$ with their assigned colors $c_v, c_w$. 
If the two measured vertices are neighbors, one can check whether or not the color constraint is satisfied on this edge.
When the two proofs are  \emph{unentangled}, Arthur can verify that the proofs have the ``correct'' form: (i) each vertex is assigned exactly one color, and (ii) the coloring is consistent across the two proofs.

The catch is that the protocol decides $\NP$ with an \emph{inverse polynomial} completeness-soundness gap.  Thus, a straightforward application of the same protocol (for a succinct version of 3-colorability \cite{papad_succinct}) decides $\NEXP$ with \emph{inverse exponential} completeness-soundness gap. This is too small a gap for $\QMA(2)$.\footnote{Intuitively, Arthur most often sees two vertices which are not connected by an edge and so he cannot check the validity of the coloring. In the $O(\log n)$-qubit version he saw adjacent vertices with an inverse-polynomial probability, but now this occurs with an inverse-exponential probability.}
What causes this small gap in the Blier-Tapp protocol?
The first source is that when the graph is not 3-colorable, there may be only \emph{one} edge that violates the coloring, so it is very difficult to sample this edge. This can be overcome by taking the PCP version of \textsc{3-coloring}. 
The second source is that for \emph{sparse} graphs, the probability of that two vertices chosen at random are adjacent is inverse polynomial in the number of vertices. 
One may try to fix this with a \emph{dense} graph; however, since one can efficiently approximate \textsc{3-coloring} on \emph{dense} graphs~\cite{ARORA1999}, we can no longer use a PCP to fix the first issue.

To simultaneously overcome both issues in the Blier-Tapp protocol and appropriately scale up to $\NEXP$, Jeronimo and Wu~\cite{jeronimo2023power}~consider a different succinct gapped constraint satisfaction problem (CSP).
Their protocol follows a similar outline.\footnote{Each quantum proof, instead of encoding each vertex and a chosen color, now encodes a each constraint and a list of assigned values to its corresponding variables.
For ease of discussion, we continue to call the registers ``vertex'' and ``color'' registers.
This is just a difference in notation, as each constraint only uses a constant number of variables.}
A key part of the protocol in~\cite{jeronimo2023power} is to enforce the same ``correct'' form of each quantum proof, as used in the Blier-Tapp protocol. This form was later named \emph{rigid} in \cite{bassirian2023qmaplus}:
\begin{definition}[Rigid and quasirigid states]
\label{defn:rigid}
A bipartite state $\ket{\psi}\in\CC^R\otimes\CC^\kappa$ is called quasirigid if it is of the form $\ket{\psi} =  \sum_{v \in [R]} \alpha_v \ket{v} \ket{c_v}$ for some complex numbers $\{\alpha_v\}_{v\in[R]}$ and indices $\{c_v\in[\kappa]\}_{v\in[R]}$, and moreover called \emph{rigid} if all $\alpha_v = \frac{1}{\sqrt{R}}$.
\end{definition}
However, \cite{jeronimo2023power} only managed to enforce the \emph{rigidity} of each quantum proof by requiring additional structure, i.e. 
the proofs are unentangled, \emph{and} have non-negative amplitudes in the computational basis.
Bassirian, Fefferman, and Marwaha~\cite{bassirian2023qmaplus} show how to use only the non-negative amplitudes requirement (i.e. without imposing unentanglement) to get the same result. 
This is done using two distinct tests: (i) check that each vertex is assigned exactly one color (i.e. \emph{quasirigid} or \emph{valid}), and (ii) check that every vertex is represented in the proof with nontrivial weight.

In this section we go the other direction: we enforce that a quantum proof is \emph{rigid} for pure states with our specific entanglement structure of \emph{internal separability}.
We closely follow the recipe of \cite{bassirian2023qmaplus, bassirian2024superposition}, adjusting some of their claims (especially \cite[Lemma 10]{bassirian2023qmaplus}) to our scenario.

\subsection{A warm-up with the computational basis}
In \cite{bassirian2024superposition}, Bassirian and Marwaha create a framework to unify the previous results. They observed that a variant of $\QMA$, where Arthur is equipped with a fantastical operation they call a ``superposition detector'', can decide $\NEXP$. This operation inputs a quantum state, and decides whether the quantum state is a computational basis element, or $\epsilon$-far from any such element. 
It empowers Arthur to enforce condition (i) (quasirigidity) --- i.e. ensuring the proof has one color per vertex.
For condition (ii), Arthur needn't any special equipment. 
He tests that all vertices are in the proof with similar weight by projecting to the uniform superposition state $\ket{+}$. While this test is simple, it comes at a price: the honest proof state passes it only with probability $1/\kappa$, strictly less than $1$. 
This creates some difficulty when choosing the probabilities of 
running each test in a way that decides the problem, but it typically can be overcome.

 superposition detector can be easily implemented given two copies of $\ket{\psi}$: measure each copy in the computational basis, and accept iff the outcomes agree.
But in~\cite{bassirian2024superposition}, this operation is applied to a single quantum proof after a partial measurement in the computational basis.
Nonetheless, \cite{bassirian2024superposition} shows how to efficiently implement this operation if the proof has a particular form, which includes states with non-negative amplitudes (showing the main result of \cite{bassirian2023qmaplus} within the framework).

Here, we warm up by considering two entanglement guarantees that are stronger than $\cpp$.
We show how to decide $\NEXP$ with each entanglement guarantee using the framework of \cite{bassirian2024superposition}.
We progressively relax the guarantees, eventually working with $\cpp$ in the next subsection.

We start with a set of tripartite proofs --- where the first register may contain $O(\poly(n))$ qubits, and the other two registers are of constant size --- that roughly corresponds to having ``two copies'' of the color register. Precisely, the set is $\ccs \defeq \{ \ket{\psi} \ |\ \ket{\psi} = \sum_i \alpha_i \ket{i}\ket{\phi_i}\ket{\phi_i}\}$. 
Thus, after measuring any outcome $\ket{i}$, the reduced state is some $\ket{\phi_i}\ket{\phi_i}$.
Since the reduced state has two copies, we can implement a superposition detector and enforce condition (i).
Thus, a quantum proof from this state must be close to $\sum_v\alpha_v\ket{v}\ket{c_v}\ket{c_v}$ (where $\ket{c_v}$ is a computational basis state).
We may convert such a state to a \emph{quasirigid} state by applying a $\cnot$ on the last two registers. From here, the framework of \cite{bassirian2024superposition} shows how to choose the test probabilities to enforce a rigid state, and thus decide $\NEXP$.

We now slightly relax the guarantee by considering the set of computational basis purifications of separable states: $\ccc=  \{ \ket{\psi} \ | \ \ket{\psi}= \sum_i \alpha_i \ket{i}\ket{\mu_i}\ket{\nu_i} \}$. The difference between $\ccs$ and $\ccc$ is that $\ccc$ contains states where $\ket{\mu_i}$ and $\ket{\nu_i}$ are different. We claim that a similar argument as before shows $\QMACP = \NEXP$.

\begin{claim}[{Inspired by \cite[Claim 21]{bassirian2024superposition}}]
\label{claim:claim21_for_vcomp}
Consider a quantum state of the form $\ket{\mu}\ket{\nu}$, where each register has $k$ qubits. For $0<\epsilon\leq1/2$, there is an efficient quantum protocol that decides if (YES) $\ket{\mu}\ket{\nu}$ is two copies of a computational basis element $\ket{e}\ket{e}$, or (NO) $\epsilon$-far in trace distance from any such basis element, with completeness and soundness gap $2x(1-x)$ for $x = \min(\epsilon, \sqrt{1-\epsilon}, \sqrt{\frac{1-\epsilon}{2^k}})$.
\end{claim}

\begin{proof}
Let $\ket{\mu} = \sum_j \beta_j \ket{j}$ and $\ket{\nu} = \sum_j \delta_j \ket{j}$, where $\sum_j |\beta_j|^2 = \sum_j |\delta_j|^2 = 1$. The protocol measures each state in the computational basis and accepts iff they agree. This succeeds with probability 
\begin{align*}
    \sum_j |\beta_j \delta_j|^2\,.
\end{align*}
In the YES case, $\ket{\mu} = \ket{\nu}$ both are equal to some computational basis element $\ket{j}$, and so the protocol succeeds with probability $1$. 

In the NO case, $|\beta_j \delta_j|^2 = |\bra{\mu,\nu}\ket{j,j}|^2 \le 1- \epsilon$ for every computational basis element $\ket{j}$. We suppose the probability of success is at least $1 - 2\epsilon$, since otherwise there is a completeness-soundness gap of $2 \epsilon$. Let $j^*$ be a value that maximizes $|\beta_j \delta_j|^2$ (break ties arbitrarily). Then $|\beta_{j^*} \delta_{j^*}|^2 \ge \frac{1-2\epsilon}{2^k}$. So the probability of success is 
\begin{align*}
    \sum_j |\beta_j \delta_j|^2 
    &= |\beta_{j^*} \delta_{j^*}|^2 + \sum_{j \ne j^*} |\beta_j \delta_j|^2 
    \\
    &\le |\beta_{j^*} \delta_{j^*}|^2 + \big( \sum_{j \ne j^*} |\beta_j | \cdot |\delta_j| \big)^2 
    \\ 
    &\le |\beta_{j^*} \delta_{j^*}|^2 + \big( \sum_{j \ne j^*} |\beta_j |^2 \big) \cdot \big(  \sum_{j \ne j^*} |\delta_j|^2 \big) 
    \\
    &\le |\beta_{j^*} \delta_{j^*}|^2 + (1 - |\beta_{j^*}|^2)\cdot (1 -  |\delta_{j^*}|^2)\,.
\end{align*}
The second inequality follows by Cauchy-Schwarz. We next invoke the AM-GM inequality: for non-negative $x,y$, we have $2\sqrt{xy} \le x + y$. This implies
\begin{align*}
    (1 - x)(1-y) = 1 - x - y + xy \le 1 - 2\sqrt{xy} + xy = {(1 - \sqrt{xy})}^2\,.
\end{align*}
So the probability of acceptance is at most $|\beta_{j^*}\delta_{j^*}|^2 + (1 - |\beta_{j^*}\delta_{j^*}|)^2 = 1 - 2|\beta_{j^*}\delta_{j^*}| \cdot (1 - |\beta_{j^*}\delta_{j^*}|)$.
This expression is maximized at the smallest and largest values of $|\beta_{j^*}\delta_{j^*}|$. By definition of the NO case, $|\beta_{j^*}\delta_{j^*}| \le \sqrt{1-\epsilon}$, and by assumption, $|\beta_{j^*}\delta_{j^*}| \ge \sqrt{\frac{1-\epsilon}{2^k}}$.\end{proof}
Since we can still implement the superposition detector with a single proof from $\ccc$, we may use the framework of \cite{bassirian2024superposition} to enforce rigidity of our proof (up to a $\cnot$ gate), and so $\QMACP = \NEXP$.

\subsection{Handling purifications in an arbitrary basis}
\label{sec:pure2polypolyconstNexp}
In $\ccc$, after a computational basis measurement in the first register, the post-measurement state is \emph{always} separable. We took advantage of this structure to implement a superposition detector in  \Cref{claim:claim21_for_vcomp}.
Unfortunately, showing that $\QMAPP = \NEXP$ is not quite as simple. 
The main difference between $\cpp$ and $\ccc$ is that in $\cpp$, the purification register can be oriented in any basis.
For example, the below state belongs to $\cpp$, but not $\ccc$:
\begin{align*}
    \tfrac{1}{\sqrt{2}}(\ket{+}\ket{0}\ket{0}+\ket{-}\ket{1}\ket{1})=\tfrac{1}{2}(\ket{0}(\ket{00}+\ket{11})+\ket{1}(\ket{00}-\ket{11}))\,.
\end{align*}
To amend this issue, we add data to the second register to mimic the behavior of states in $\ccc$.

Consider states on four registers $\mathbb{C}^R \otimes (\mathbb{C}^\kappa \otimes \mathbb{C}^R) \otimes \mathbb{C}^\kappa$, where $R = 2^{\poly(n)}$ and $\kappa$ is some constant. Let $p_1, p_2$ be probabilities that we set later. Our protocol is as follows:
\begin{itemize}
    \item \emph{Enforce ``quasirigidity''.} With probability $p_1$ run $\semicheck$ --- measure all the registers in the computational basis, and make sure it is of the form $\ket{v}\ket{c}\ket{v}\ket{c}$ (i.e.\ the first pair of outcomes \emph{matches} the second pair of outcomes).
    \item \emph{Enforce ``density''.} With probability $p_2$ run $\cnot_{1,3}$ (i.e.\ applied to the first and third register), then $\cnot_{2,4}$, and check $\density$ --- measure the first two registers in the Hadamard ($X$) basis and accept if all the qubits are $\ket{+}$.
    \item \emph{Estimate CSP value.} Otherwise run $\cnot_{1,3}$ and $\cnot_{2,4}$, and then run the constraint tests of~\cite{bassirian2023qmaplus} on the first two registers, which check that the encoded classical assignment satisfies the CSP.
\end{itemize}
Why does this work? It turns out that $\semicheck$ combined with the internal separability guarantee enforces the state (after applying CNOTs) to be close to \emph{quasirigid}. Adding the $\density$ check (after CNOTs) enforces it to be close to \emph{rigid}. However, to show this, we must reprove adjusted versions of lemmas from \cite{bassirian2023qmaplus,bassirian2024superposition}.

Consider an arbitrary tripartite state
\begin{align}
\label{eqn:arbitrary_tripartite_state}
    \ket{\psi} \defeq \sum_{v,c,w,d}  a_{vc,wd} \ket{v}_A\ket{c,w}_B\ket{d}_C\,,
\end{align}
and let $\rho_{B, C} = \tr_A(\ketbra{\psi})$ be the reduced density matrix after tracing out register $A$. Then the success probability of $\semicheck$ is exactly $\sum_{vc} |a_{vc,vc}|^2$.

Now consider $\ket{\psi}$ after applying $\cnot_{1,3}$ and $\cnot_{2,4}$:
\begin{align*}
        \cnot_{1,3} \cnot_{2,4} \ket{\psi} = \sum_{vc,wd}  a_{vc,wd} \ket{v}\ket{c}\ket{w \oplus v}\ket{d \oplus c}\,.
\end{align*}
The last two registers are $\ket{0}\ket{0}$ exactly when $(v,c) = (w,d)$, which happens with \emph{exactly} the probability that $\semicheck$ succeeds. So conditioned on the last two registers being $\ket{0}\ket{0}$, the new state is proportional to 
\begin{align*}
 \left(\sum_{vc} a_{vc,vc}\ket{v}\ket{c}\right)\ket{0}\ket{0}\,.
\end{align*}
It remains to show that the new state is quasirigid.
Suppose $\semicheck$ passes with high probability; then for a typical $\ket{v}_A$ in state \cref{eqn:arbitrary_tripartite_state}, the two color registers are the roughly in superposition of $\ket{cc}$ for allowed colors $c$. If the superposition consists more than one color, it adds some entanglement between the $B$ and $C$ registers.
Since we have that state $\rho_{B,C}$ is \emph{separable}, then $\ket{\psi}$ must be close to quasirigid.

To prove this claim formally, we require a technical lemma:
\begin{lemma}
\label{lemma:separable_offdiagonal_cauchy_schwarz}
    Consider a state $\rho \in \cd(\mathbb{C}^m, \mathbb{C}^n)$ that is separable. Then for any $w,y \in [m]$ and $x,z \in [n]$, we have
    \begin{align*}
        |\bra{wx}\rho\ket{yz}| \le \sqrt{\min(\bra{wx}\rho\ket{wx}, \bra{wz}\rho\ket{wz}, \bra{yx}\rho\ket{yx}, \bra{yz}\rho\ket{yz})}\,,
    \end{align*}
    and moreover,
    \begin{align*}
        |\bra{wx}\rho\ket{yz}| \le \frac{1}{2} \cdot \min(\bra{wx}\rho\ket{wx}+ \bra{yz}\rho\ket{yz}, \bra{wz}\rho\ket{wz}+ \bra{yx}\rho\ket{yx})\,.
    \end{align*}
\end{lemma}
\begin{proof}
    Fix a separable decomposition of $\rho$ of the form $\rho = \sum_\iota p_\iota M_\iota \otimes N_\iota$, where all $p_\iota \ge 0$, $\sum_\iota p_\iota = 1$, and $M_\iota, N_\iota$ are trace-1 and positive semidefinite. Choose any $w,y \in [m]$ and $x,z \in [n]$. Then
    \begin{align*}
        \left|\bra{wx}\rho\ket{yz}\right| &\le \left| \sum_\iota p_\iota \bra{w}M_\iota\ket{y}\cdot \bra{x}N_\iota\ket{z} \right|
        \le  \sum_\iota p_\iota \left| \bra{w}M_\iota\ket{y} \right| \cdot \left| \bra{x}N_\iota\ket{z} \right|\,.
    \end{align*}
    Recall that PSD matrices $A$ decompose as $A = B^\dagger B$, and so by Cauchy-Schwarz,
    \begin{align*}
            |\bra{i}A\ket{k}|^2 
            = \left| (B\ket{i})^\dagger (B\ket{k}) \right|^2 
            \le  \bra{i}B^\dagger B\ket{i} \cdot \bra{k}B^\dagger B\ket{k}
            = \bra{i}A\ket{i} \cdot \bra{k}A\ket{k}\,,
    \end{align*}    
    So then $\left|\bra{wx}\rho\ket{yz} \right| \le \sum_\iota p_\iota  \sqrt{ 
    \bra{w} M_\iota \ket{w} \cdot \bra{y} M_\iota \ket{y}
    \cdot 
    \bra{x} N_\iota \ket{x} \cdot \bra{z} N_\iota \ket{z}
    }$.

    For any PSD matrix $A$ and vector $\ket{\psi}$, we have $0 \le \bra{\psi}A\ket{\psi} \le \Tr(A)$. All $M_\iota, N_\iota$ are trace-1, so
    \begin{align*}
        \left|\bra{wx}\rho\ket{yz}\right| \le \sum_\iota p_\iota \sqrt{ 
    \bra{w} M_\iota \ket{w} \cdot \bra{z} N_\iota \ket{z}
    }
    \le
    \sqrt{\sum_\iota p_\iota \bra{w} M_\iota \ket{w} \cdot \bra{z} N_\iota \ket{z}}
    = \sqrt{\bra{wz}\rho\ket{wz}}\,,
    \end{align*}
    where the second inequality is Jensen's.
    The other bounds follow mutatis mutandis.

Moreover, we can use the AM-GM inequality; i.e. for non-negative $a,b$, we have $\sqrt{ab} \le \frac{a + b}{2}$. So
    \begin{align*}
        \left|\bra{wx}\rho\ket{yz} \right| &\le
        \frac{1}{2} \sum_\iota p_\iota \left( \bra{w} M_\iota \ket{w} \cdot \bra{x} N_\iota \ket{x} 
        +  \bra{y} M_\iota \ket{y}
    \cdot 
    \bra{z} N_\iota \ket{z}\right)
    \\
    &= \frac{\bra{wx}\rho\ket{wx} + \bra{yz}\rho\ket{yz}}{2}\,.
    \end{align*}
    Similarly, $\left|\bra{wx}\rho\ket{yz} \right| \le \frac{1}{2} \cdot (\bra{wz}\rho\ket{wz} + \bra{yx}\rho\ket{yx})$.
\end{proof}

\begin{lemma} \label{lemma:semicheck_means_close_to_quasirigid}
        Suppose $\semicheck$ succeeds with probability $1-\epsilon$, and $\rho_{B, C}$ is separable. Then there is a state $\ket{\phi}$ that is \emph{quasirigid} after applying $\cnot_{1,3}\cnot_{2,4}$ such that $|\bra{\phi}\ket{\psi}|^2 \ge 1 - (\kappa + 1) \epsilon$.\footnote{Recall that $\kappa$ is the dimension of the third register. Although this lemma resembles~\cite[Lemma 18]{bassirian2024superposition}, note that $\epsilon$ has a different meaning in this paper.}
\end{lemma}
\begin{proof}
Label the computational basis elements of $\ket{\psi}$ as $a_{vc,wd}\ket{v}\ket{c,w}\ket{d}$ as in \cref{eqn:arbitrary_tripartite_state}.
Let $\sigma: [R] \to [\kappa]$ be a function that picks out the largest amplitude per vertex, i.e. where $|a_{v\sigma(v), v\sigma(v)}|^2 \ge |a_{vd, vd}|^2$ for all $d \ne \sigma(v)$. For convenience, denote $\sigma_v \defeq \sigma(v)$.

Consider the matrix elements of the reduced density matrix of $\ket{\psi}$ after tracing out the first register:
\begin{align*}
        \bra{c, v, d}\rho_{B, C}\ket{e, w, f} = \sum_x a_{xc,vd} a_{xe,wf}^\dagger\,.
\end{align*}

We use this expression to upper-bound the weight on colors not chosen by $\sigma$:
\begin{align*}
\sum_v \sum_{d;d \ne \sigma_v}    |a_{vd, vd}|^2 
&\le
\sum_v \sum_{d;d \ne \sigma_v}    |a_{v\sigma_v, v\sigma_v}| \cdot |a_{vd, vd}^\dagger |
\\
&=
\sum_v \sum_{d;d \ne \sigma_v}    |a_{v\sigma_v, v\sigma_v} \cdot a_{vd, vd}^\dagger |
\\
&=
\sum_v \sum_{d;d \ne \sigma_v}  
\left(
        |\bra{\sigma_v, v, \sigma_v}\rho_{B, C}\ket{d, v, d} - \sum_{x;x \ne v} a_{x\sigma_v,v\sigma_v} \cdot a_{xd,vd}^\dagger|
\right)
\\
&\le 
\Bigg(
 \sum_v \sum_{d;d \ne \sigma_v} 
|\bra{\sigma_v, v, \sigma_v}\rho_{B, C}\ket{d, v, d}| 
\Bigg)
+ \sum_v \sum_{d;d \ne \sigma_v}  
 \sum_{x;x \ne v} |a_{x\sigma_v,v\sigma_v}| \cdot |a_{xd,vd}^\dagger|
\,.
\end{align*}
The second term is actually quite small. By the AM-GM inequality,
\begin{align*}
    |a_{x\sigma_v,v\sigma_v}| \cdot |a_{xd,vd}^\dagger| \le \frac{1}{2} \cdot \left(|a_{x\sigma_v,v\sigma_v}|^2 +  |a_{xd,vd}^\dagger|^2 \right)\,.
\end{align*}
Taking $|a_{xd,vd}|^2$ over all three sums $(v,d,x)$ is at most $\sum_{(x,c) \ne  (v,d)} |a_{xc,vd}|^2 = 1 - (1-\epsilon) = \epsilon$. 
Similarly, summing $|a_{x\sigma_v,v\sigma_v}|^2$ over sums $(v,x)$ contributes at most $\epsilon$, so at most $\kappa \epsilon$ over sums $(v,d,x)$.
In total, the second term is at most $\frac{\kappa+1}{2} \epsilon$.

For the first term, since $\rho_{B, C}$ is separable, we can use one of the bounds proven in \Cref{lemma:separable_offdiagonal_cauchy_schwarz}, i.e. 
\begin{align*}
        |\bra{\sigma_v, v, \sigma_v}\rho_{B, C} \ket{d, v, d}| \le \frac{1}{2} \cdot \big(\bra{\sigma_v, v, d}\rho_{B, C}\ket{\sigma_v, v, d} + \bra{d, v, \sigma_v}\rho_{B, C}\ket{d, v, \sigma_v}\big)\,.
\end{align*}
So then
\begin{align*}
    \sum_v \sum_{d;d \ne \sigma_v}    |a_{vd, vd}|^2 
    &\le
    \frac{\kappa+1}{2} \epsilon + 
    \frac{1}{2} \left(\sum_v \sum_{d;d \ne \sigma_v} 
     \bra{\sigma_v, v, d}\rho_{B, C}\ket{\sigma_v, v, d} + \bra{d, v, \sigma_v}\rho_{B, C}\ket{d, v, \sigma_v}\right)
    \\
    &\le \frac{\kappa+1}{2} \epsilon + 
     \frac{1}{2} \left(\sum_v \sum_{d;d \ne \sigma_v} 
     \sum_{x} |a_{x\sigma_v,vd}|^2 
     + |a_{xd,v\sigma_v}|^2\right)\,.
\end{align*}
The remaining sum is at most $\sum_{(x,c) \ne  (v,d)} |a_{xc,vd}|^2 = \epsilon$, so the upper bound is at most $\frac{\kappa + 2}{2}\epsilon \le \kappa \epsilon$.

Consider the state $\ket{\phi} = \frac{1}{\sqrt{\gamma}} \sum_v a_{v \sigma_v, v\sigma_v} \ket{v}\ket{\sigma_v}\ket{v}\ket{\sigma_v}$, where $\gamma = \sum_v |a_{v \sigma_v, v\sigma_v}|^2$ is chosen so that $\braket{\phi}{\phi} = 1$. Note that $\ket{\phi}$ is quasirigid after applying $\cnot_{1,3}$ and $\cnot_{2,4}$. Moreover, $|\braket{\phi}{\psi}|^2 = \gamma$. Putting everything together,
\begin{align*}
    \gamma = \sum_v |a_{v \sigma_v, v\sigma_v}|^2 = \Pr[\semicheck \text{ succeeds}] -  \sum_v \sum_{d;d \ne \sigma_v}    |a_{vd, vd}|^2  \ge 1 - \epsilon - \kappa \epsilon\,.\tag*{\qedhere} 
\end{align*}
\end{proof}

We also use the following fact in several places:
\begin{fact}[\cite{fuchs1998cryptographic}]
\label{fact:fuchs_vdg_implication}
    Let $0 \le \Pi \le \mathbb{I}$ be a positive semi-definite matrix, and let $\ket{\psi_1}$ and $\ket{\psi_2}$ be quantum states such that $|\braket{\psi_1}{\psi_2}|^2 \ge 1 - d$. Then $|\bra{\psi_1}\Pi\ket{\psi_1} - \bra{\psi_2}\Pi\ket{\psi_2}| \le \sqrt{d}$.
\end{fact}

We are now ready to show that a quantum proof passing both tests must be close to quasirigid:
\begin{claim}
\label{claim:rigidity}
    Suppose $\density$ and $\semicheck$ succeed on $\ket{\psi}$ with probability $\frac{1}{\kappa} - d_{\mathsf{D}}$ and $1- d_\mathsf{M}$, respectively. Then there exists a state $\ket{\chi}$ that is rigid after applying $\cnot_{1,3}\cnot_{2,4}$ such that 
    \begin{align*}
    |\braket{\chi}{\psi}|^2 &\ge 1 - \kappa d_{\mathsf{D}}  - (\kappa+1) \sqrt{(\kappa + 1) d_\mathsf{M}}\,.
    \end{align*}
\end{claim}
\begin{proof}
    By \Cref{lemma:semicheck_means_close_to_quasirigid}, there exists a state $\ket{\phi}=\sum_v\alpha_v\ket{v}\ket{c_v,v}\ket{c_v}$ such that 
    \begin{align*}
        |\braket{\psi}{\phi}|^2 \ge 1 - (\kappa+1)d_\mathsf{M}\,.
    \end{align*}
    By \Cref{fact:fuchs_vdg_implication}, 
    for any  unit vector $\ket{\mu}$,
    we have
    $\left| |\braket{\mu}{\psi}|^2 - |\braket{\mu}{\phi}|^2 \right| \le \sqrt{(\kappa+1)d_\mathsf{M}}$.
    We use this in two places. 
    Let $\ket{+'} = \cnot_{1,3}\cnot_{2,4}\ket{+}\ket{0}\ket{0}$. First, since  $|\braket{+'}{\psi}|^2 \ge \frac{1}{\kappa} - d_{\mathsf{D}}$, applying \Cref{fact:fuchs_vdg_implication} with $\ket{\mu} = \ket{+'}$ implies $|\braket{+'}{\phi}|^2 \ge \frac{1}{\kappa} - d_{\mathsf{D}} -\sqrt{(\kappa+1)d_\mathsf{M}}$. 
    Now let $\ket{\chi}=\frac{1}{\sqrt{R}}\sum_v \ket{v}\ket{c_v,v}\ket{c_v}$ be the state with the same computational basis elements as $\ket{\phi}$. Then
    \begin{align*}
        |\braket{\chi}{\phi}|^2 =\left|\frac{1}{\sqrt{R}}\sum_{v\in[R]}\alpha_v \right|^2
        =\kappa \left|\frac{1}{\sqrt{\kappa R}} \sum_{v\in[R]}^R \alpha_v\right|^2
        = \kappa |\braket{+'}{\phi}|^2
        \ge 1 - \kappa\left(d_{\mathsf{D}} +\sqrt{(\kappa+1)d_\mathsf{M}}\right)\,.
    \end{align*}
   Applying \Cref{fact:fuchs_vdg_implication} using $\ket{\mu} = \ket{\chi}$ then implies
   \begin{align*}
       |\braket{\chi}{\psi}|^2  \ge  1 - \kappa\left(d_{\mathsf{D}} +\sqrt{(\kappa+1)d_\mathsf{M}}\right) - \sqrt{(\kappa+1)d_\mathsf{M}}=1 - \kappa d_{\mathsf{D}}  - (\kappa+1) \sqrt{(\kappa + 1) d_\mathsf{M}}\,.\tag*{\qedhere}
   \end{align*}
\end{proof}
Recall that in the YES case, the $\density$ test succeeds with probability strictly less than $1$. A conniving Merlin in the NO case could try to take advantage of this fact.
To mitigate this issue, 
we show that if the $\density$ test succeeds with probability higher than in the YES case, the success of the $\semicheck$ test is penalized.
We prove a tradeoff between the success of these two tests:
\begin{lemma}
\label{lemma:quadratic}
    Suppose $\density$ and $\semicheck$ succeed on state $\ket{\psi}$ with probability $w_{\mathsf{D}} \ge \frac{1}{\kappa}$ and $w_{\mathsf{M}}$, respectively. Then $(w_{\mathsf{D}} -  \frac{1}{\kappa})^2 +  (\kappa+1)w_{\mathsf{M}}\le  \kappa+1$.
\end{lemma}
\begin{proof}
By \Cref{lemma:semicheck_means_close_to_quasirigid}, there exists a $\ket{\phi} = \sum_{v \in [R]} \alpha_v \ket{v}\ket{c_v,v}\ket{c_v}$ such that  $|\braket{\psi}{\phi}|^2 \ge 1-(\kappa+1)(1 - w_{\mathsf{M}})$. 
Using \Cref{fact:fuchs_vdg_implication}, we see that 
$\left| |\braket{+'}{\psi}|^2 - |\braket{+'}{\phi}|^2 \right|$ is at most $\sqrt{(\kappa+1)(1-w_{\mathsf{M}})}$. Note that by assumption, $|\braket{+'}{\psi}|^2 = w_{\mathsf{D}}\ge\frac{1}{\kappa}$, and by Cauchy-Schwarz, $|\braket{+'}{\phi}|^2 = \frac{1}{R \cdot \kappa} \left| \sum_{v \in [R]} \alpha_v \right|^2 \le \frac{1}{\kappa} \sum_{v \in [R]} |\alpha_v|^2 = \frac{1}{\kappa}$. Combining these claims, we have 
\begin{align*}
    w_{\mathsf{D}} - \frac{1}{\kappa} \le |\braket{+'}{\psi}|^2 - |\braket{+'}{\phi}|^2 \le  \sqrt{(\kappa+1)(1-w_{\mathsf{M}})}\,.
\end{align*}
The lemma follows after rearranging the terms.
\end{proof}

After following the recipe in \cite{bassirian2024superposition} and adjusting many tools from there, we have in \Cref{claim:rigidity} and \Cref{lemma:quadratic} two tests that enforce \emph{rigidity} of the quantum proof. These statements allow us to set the probability of running each test (i.e. $p_1,p_2$) in the protocol. We defer the full calculation to \Cref{sec:probabilities}. Once this is done, we reach our goal in this section:
\begin{theorem}\label{thm:qmapp-nexp}
     There exist constants $1 > c > s > 0$ such that $\QMAPP^{c, s} = \NEXP$.
\end{theorem}

\section{Two copies enforce purity: $\QMAPP \subseteq \QMAP(2)$}
\label{sec:rank1restrict}
We can use two unentangled proofs to implement a SWAP test. For product states, this ensures the unentangled proofs are \emph{pure}.
Consider a convex set of quantum states $\cw$ with the property that nearly-pure states in $\cw$ are always close to a state in $\purerestrict(\cw)$, i.e. a pure state in $\cw$.
When this is the case, we show that any protocol in $\QMA_{\purerestrict(\cw)}$ can be simulated in $\QMA_{\cw}(2)$.

\begin{definition}[Continuity condition of $\cw$]\label{def:continuityProp}
A set of quantum states $\cw$ has the \emph{continuity condition} if there exists a number $0 < a\le 1$ and a continuous, strictly monotonic function $f: [0,a] \to [0,1]$ with $f(0) = 0$ satisfying the following property: For any $\rho \in \cw$, if the maximum eigenvalue of $\rho$ is $1 - x$, then there is some $\ket{\psi} \in \purerestrict(\cw)$
where $\bra{\psi}\rho\ket{\psi} \ge 1 - f(x)$.
\end{definition}
The continuity condition guarantees that $\rho \in \cw$ is $\epsilon$-close to a pure state only if it is $f(\epsilon)$-close to a pure state \emph{in $\cw$}. Although this condition is satisfied by many sets $\cw$, for example the set of separable states, there exist pathological sets without this condition.\footnote{For example, let $P_d\subseteq \cd(\CC^d)$ 
be the set of states with overlap at least $1-\tfrac1d$ with some pure state, and define $\mathcal{G}\defeq\bigcup_d (\cd(\CC^d)\setminus P_d)$.
Then $\mathcal{G}$ has states that are arbitrarily close to pure, yet $\purerestrict(\mathcal{G})$ is empty.
} 

\begin{lemma}
\label{lemma:swaptest_rank1restrict}
    Consider any protocol $V$ in $\QMA_{\purerestrict(\cw)}$ with completeness-soundness gap $\Delta$.
    Suppose also that $\cw$ has the continuity condition with function $f$.
    Then $\QMA_{\cw}(2)$ can simulate $V$ with completeness-soundness gap $\Omega(\Delta \cdot f^{-1}(\frac{\Delta^2}{4}))$.
\end{lemma}
\begin{proof}
The quantum proof in $\QMA_{\cw}(2)$ is some product state $\rho_1 \otimes \rho_2$ where $\rho_1, \rho_2 \in \cw$. The verifier runs the following procedure, where we fix the value of $p$ later on:
    \begin{itemize}
        \item With probability $p$, run the SWAP test.
        \item Otherwise, choose $i \in \{1,2\}$ uniformly at random and run $V$ on $\rho_i$.
    \end{itemize}
The SWAP test passes with probability $\frac{1}{2}(1 + \Tr[\rho_1\rho_2])$. Notice that if this is $1$, then $\rho_1 = \rho_2$ is a pure state; i.e. $\rho_1,\rho_2 \in \purerestrict(\cw)$. In completeness, Merlin sends an honest quantum proof, so the procedure accepts with probability at least
$p + (1-p)(s + \Delta)$, where $s$ is the soundness guaranteed by $V$.
We'll split the analysis of soundness into cases depending on the maximum overlap of $\rho_1$ and $\rho_2$ with a pure state.

Let $\rho_1 = \sum_i p_i \ketbra{\psi_{i}}$ and $\rho_2 = \sum_j q_j \ketbra{\phi_j}$ where $\{\ket{\psi_i}\}$ and $\{\ket{\phi_j}\}$ are orthonormal sets, and the $p_i$'s and $q_i$'s are in decreasing order. Suppose $p_1 q_1 = 1 - \epsilon$. Then each of $p_1, q_1$ is at least $1 - \epsilon$. 

\begin{enumerate}
    \item The first case is when $\epsilon \ge 1/3$. Here, the proofs are very far from pure, and so the SWAP test fails. Notice that $\Tr[\rho_1 \rho_2] = \sum_j q_j \bra{\phi_j}\rho_1\ket{\phi_j} \le \max_{\ket{\chi}} \bra{\chi}\rho_1 \ket{\chi} = p_1$. Similarly, $\Tr[\rho_1 \rho_2] \le q_1$.  
    So then $\Tr[\rho_1 \rho_2]^2 \le \min(p_1, q_1)^2 \le p_1 q_1 = 1 - \epsilon$.
    Thus, the SWAP test succeeds with probability $\frac{1}{2}(1 + \Tr[\rho_1 \rho_2]) \le \frac{1}{2}(1 + \sqrt{2/3})$. So the completeness-soundness gap is at least
    \begin{align*}
    p + (1-p)&(s + \Delta) - \left( p \cdot \frac{1}{2}(1 + \sqrt{2/3}) + (1-p) \right)
    \\
    &=
    p \cdot \frac{1}{2}(1 - \sqrt{2/3}) + (1-p)(s + \Delta - 1)
    \,.
    \end{align*}
    The gap is at least $\delta > 0$ when the probability $p$ is at least 
    \begin{align}
    \label{eqn:condition1}
        p > \frac{1 - (s + \Delta) + \delta}{1 - (s + \Delta) + \frac{1}{2}(1 - \sqrt{2/3})}\,.
    \end{align}
    \item Otherwise, $\epsilon < 1/3$. We will need the probability of the SWAP test in this case, whose proof is deferred:
    \begin{claim}\label{claim:mixed-swap}
        When $\epsilon < 1/3$, the SWAP test accepts with probability at most $1 - \frac{1}{2} \epsilon + \frac{1}{2} \epsilon^2$. 
    \end{claim}
We next analyze the acceptance probability of the other test. Suppose $i=1$; the $i=2$ case is the same by symmetry.
Let $f$ be the polynomial associated with the continuity condition of $\cw$. Since $p_1 \ge 1 - \epsilon$, $\rho_1$ has squared overlap at least $1-f(\epsilon)$ with a state in $\purerestrict(\cw)$. Since one state is pure, the squared overlap is equal to the fidelity.\footnote{We use the ``squared'' version of fidelity, i.e. $F(\rho, \sigma) \defeq \left(\Tr[\sqrt{\sqrt{\rho} \sigma \sqrt{\rho}}] \right)^2$.}
Fuchs-van de Graaf~\cite{fuchs1998cryptographic} implies that the trace distance is at most $\sqrt{f(\epsilon)}$. 
So this test accepts with probability at most $s + \sqrt{f(\epsilon)}$.

Altogether, the procedure accepts with probability at most
\begin{align*}
    p \cdot (1 - \frac{1}{2}\epsilon + \frac{1}{2}\epsilon^2) + (1-p)(s + \sqrt{f(\epsilon)})\,,
\end{align*}
and so the completeness-soundness gap is at least
\begin{align*}
    \frac{1}{2} p \cdot \epsilon(1 - \epsilon) + (1-p)(\Delta - \sqrt{f(\epsilon)})\,.
\end{align*}
Choose $\gamma>0$ where $f(\gamma) = \frac{\Delta^2}{4}$.\footnote{If this is not possible for any $\gamma \le 1/3$, then set $\gamma = 1/3$, and ignore case (b).} We further divide into two subcases:
\begin{enumerate}
    \item When $\epsilon \le \gamma$, the proof is close to a state in $\purerestrict(\cw)$, so the gap is at least $(1-p)\frac{\Delta}{2}$.
    \item Otherwise, $\gamma   < \epsilon  < 1/3$. Then the proofs are far from pure and fail the SWAP test, even if $f(\epsilon) = 1$. The gap is at least $\delta > 0$ when the probability $p$ is at least
\begin{align}
\label{eqn:condition2}
    p > \frac{1 - \Delta + \delta}{1 - \Delta  + \frac{1}{2} \gamma(1 - \gamma)}\,.
\end{align}
\end{enumerate}
\end{enumerate}
If we set $\delta = \frac{1}{4} \min(1 - \sqrt{2/3}, \gamma(1-\gamma))$, then we may choose $0 < p < 1$ that satisfies both \cref{eqn:condition1,eqn:condition2}, producing a completeness-soundness gap $\Omega(\min(\delta, (1-p)\Delta)) = \Omega(\delta \Delta)$. So the completeness-soundness gap is $\Omega(\Delta \cdot f^{-1}(\frac{\Delta^2}{4}))$.
\end{proof}
We finish the proof of~\cref{lemma:swaptest_rank1restrict} by proving~\cref{claim:mixed-swap}.
    \begin{proof}[Proof of~\cref{claim:mixed-swap}]
        Recall that  $\Tr[\rho_1 \rho_2] = \sum_{i,j} p_i q_j |\bra{\psi_i}\ket{\phi_j}|^2$. Then
        \begin{align*}
        \Tr[\rho_1 \rho_2]  \le p_1 q_1 |\bra{\psi_1}\ket{\phi_1}|^2 
        + \sum_{j \ne 1} q_j |\bra{\psi_1}\ket{\phi_j}|^2
+ \sum_{i \ne 1} p_i |\bra{\psi_i}\ket{\phi_1}|^2
+  \left(\sum_{i \ne 1} p_i\right)\left(\sum_{j \ne 1} q_j\right) \,.
    \end{align*}
    Let $|\bra{\psi_1}\ket{\phi_1}|^2 =1-z$.
    We simplify each term in this upper bound:
    \begin{itemize}
    \item  The first term is $(1 - \epsilon)(1-z)$.
    \item The value $|\bra{\psi_1}\ket{\phi_j}|^2$ when $j \ne 1$ is at most $1 - |\bra{\psi_1}\ket{\phi_1}|^2 = z$. So the second term is at most $z(\sum_{j \ne 1} q_j) \le z \epsilon$.
\item By symmetry, the third term is at most $z \epsilon$.
\item  Since $p_1,q_1 \ge 1 - \epsilon$, the last term is at most $\epsilon^2$.
\end{itemize}
Altogether, we have
\begin{align*}
    \Tr[\rho_1 \rho_2] \le (1 - \epsilon)(1-z) + 2 z \epsilon +\epsilon^2 = 1-\epsilon + (3\epsilon - 1)z + \epsilon^2\,.
\end{align*}
Since $\epsilon < 1/3$, this upper bound is maximized at $z = 0$. So the SWAP test accepts with probability $\frac{1}{2}(1 + \Tr[\rho_1 \rho_2]) \le 1 - \frac{1}{2}\epsilon + \frac{1}{2}\epsilon^2$.
    \end{proof}

We now use \Cref{lemma:swaptest_rank1restrict} to prove that $\QMAPP \subseteq \QMAP(2)$.
\begin{claim}
\label{cor:qmapurif2_is_nexp}
    $\QMAPP \subseteq \QMAP(2)$, and so $\QMAP(2) = \NEXP$.
\end{claim}
\begin{proof}
Ideally, we would like to show that $\cm$ has the continuity condition.
However, we are unable to do this when the first register in the proof has too few qubits. Nevertheless, we show that a subset of $\cm$ has the continuity condition, and moreover, the protocol for $\NEXP$ (\Cref{sec:qmaw_equals_nexp}) can use quantum proofs in this subset of $\cm$. Since the $\NEXP$-hard problem is complete for $\QMAPP$, every problem in $\QMAPP$ can be decided in $\QMAP(2)$.

Label the registers of the quantum proof $A,B,C$, and denote the dimension of any register $X$ as $\dim(X)$.
Let $\cm^\star \subseteq \cm$ be the set of tripartite states in $\cm$ where the registers also satisfy $\dim(A)\ge\dim(B)\cdot\dim(C)$. We show in \Cref{sec:padding}  that the protocol of \Cref{sec:qmaw_equals_nexp} can use a quantum proof in $\cpp$ obeying this register size inequality; thus,  $\NEXP \subseteq \QMA_{\purerestrict(\cm^\star)}$.

We now show that the continuity condition holds for $\cm^\star$. Consider any such quantum proof $\rho$. Suppose the largest eigenvalue-eigenvector pair of $\rho$ is $(\lambda_1, \ket{\psi})$. 
The fidelity of $\rho$ and $\ket{\psi}$ is exactly $\bra{\psi}\rho\ket{\psi} = \lambda_1$.
Since fidelity is non-decreasing under CPTP maps, the fidelity of $\Tr_A[\rho]$ and $\Tr_A[\ketbra{\psi}]$ is at least $\lambda_1$. 
Since $\ket{\psi}$ purifies $\Tr_A[\ketbra{\psi}]$, and $\dim(A) \ge \dim(B)\cdot \dim(C)$, Uhlmann's theorem implies the existence of a purification $\ket{\chi_{\rho}}$ of $\Tr_A[\rho]$ that maintains fidelity, i.e. $|\bra{\psi}\ket{\chi_{\rho}}|^2 \ge \lambda_1$.

Note that by construction, $\ket{\chi_{\rho}} \in \purerestrict(\cm^\star)$. Moreover,
\begin{align*}
    \bra{\chi_{\rho}} \rho \ket{\chi_{\rho}} \ge \lambda_1 \bra{\chi_{\rho}} \ketbra{\psi}\ket{\chi_{\rho}} \ge \lambda_1^2\,.
\end{align*}
We verify that $\cm$ has the continuity condition with $f(x) = 2x$. Suppose $\lambda_1 = 1 - x$; then the squared overlap with $\ket{\chi_{\rho}}$ is at least $(1-x)^2 = 1 - (2x - x^2) \ge 1 - 2x$. \Cref{lemma:swaptest_rank1restrict} then implies that $\NEXP \subseteq \QMA_{\cm^{\star}}(2)$, which is trivially contained in $\QMA_{\cm}(2) = \QMAP(2)$.
\end{proof}
We show in \Cref{sec:gapamplification} that $\QMAP(2)$ at large enough completeness-soundness gap can be simulated by $\QMA(2)$. As a consequence:
\begin{corollary}
    Gap amplification of $\QMAP(2)$ implies $\QMA(2) = \NEXP$.
\end{corollary}
However, by \Cref{thm:upperbound}, the single-copy variant  $\QMAP \subseteq \EXP$ at all completeness-soundness gaps at least inverse exponential in input size. Unlike $\QMA^+$, gap amplification of $\QMAP$ does not obviously imply any complexity collapse.
In other words, $\QMAP(2)$ has exactly the desired features of \cite{jeronimo2023power} while bypassing the issue identified in \cite{bassirian2023qmaplus}.

\section{Open questions}
\label{sec:outlook}
\begin{enumerate}
    \item Does the complexity class $\QMAP$ exhibit gap amplification? This is a good starting point to study gap amplification of $\QMAP(2)$, which by \Cref{cor:gapamp_means_qma2_nexp} implies $\QMA(2) = \NEXP$.
    \item The proof of \Cref{thm:lowerbound} uses two quantum proofs only to implement a purity test. Is there a set of proofs $\cw$ such that $\QMA_\cw(2)$ is strictly more powerful than $\QMA_{\purerestrict(\cw)}$? If so, there is a setting where \emph{unentanglement} is useful for something beyond purity testing.
    \item Is there a reduction from $\wval{}$ to $\wmem{}$ that preserves a constant gap? If so, we can convert our results on hardness of approximation to the associated testing problems. For example, this would imply $\NP$-hardness of \emph{testing} whether a state is a purification of a separable state or $\Omega(1)$-far from any such state.
    \item How essential is the continuity condition (\Cref{def:continuityProp}) to show $\QMA_{\purerestrict(\cw)} \subseteq \QMA_{\cw}(2)$? For example, does all of $\cm$ satisfy the continuity condition? A positive answer would simplify the proof of \Cref{cor:qmapurif2_is_nexp} by bypassing the argument in \Cref{sec:padding}. Is there a simple characterization of sets which do \emph{not} satisfy the continuity condition?
    \item One problem related to $\QMA$ and purity testing is the \emph{pure state marginal problem}. Here, one decides if there exists a \emph{pure} state that is consistent with a given set of reduced density matrices. It is known to be in $\QMA(2)$ and at least $\QMA$-hard, whereas the mixed variant is $\QMA$-complete~\cite{liu06cldm,Broadbent_2022}. Can one exactly characterize the complexity of this problem?
    \item A \emph{disentangler} is a hypothetical quantum channel 
    that maps the set of states to separable states and is approximately surjective. As noted by John Watrous (and communicated in~\cite{power_of_unentanglement}), the existence of an efficient disentangler implies $\QMA = \QMA(2)$. We remark that an analogous map from states (or even separable states) to \emph{internally separable} states would imply $\QMA(2) = \NEXP$.
    Does such a channel exist?
    In fact, disentanglers that separate $O(1)$ qubits from the rest of the state exist by the quantum de Finetti theorem~\cite{Caves_2002,Christandl_2007}, but they separate a register from the entire system, not a particular subsystem.
    \item Our work initiates a complexity-theoretic study of a notion of \emph{multipartite} unentanglement. 
    As asked in \Cref{question:multipartite_unentanglement}, which notions are computationally powerful?
    For example, what is the power of $\QMA_\cw$ when $\cw$ contains states with a multipartite entanglement structure distinct from internal separability?
\end{enumerate}

\subsection*{Acknowledgements}
KM thanks Tyler Chen, Christopher Musco, and Apoorv Vikram Singh for helpful discussions on convex optimization.
RB, KM, and IL thank the Institute for Advanced Study for organizing the summer school where they met and began collaborating.
IL and PW thank Thomas Vidick for suggesting them to talk. 

This work was done in part while a subset of the authors were visiting the Simons Institute for the Theory of Computing, supported by NSF QLCI Grant No. 2016245.
BF, RB, and KM acknowledge support from AFOSR
(FA9550-21-1-0008). This material is based upon work partially
supported by the National Science Foundation under Grant CCF-2044923
(CAREER), by the U.S. Department of Energy, Office of Science,
National Quantum Information Science Research Centers (Q-NEXT) and by
the DOE QuantISED grant DE-SC0020360. KM acknowledges support from the National Science Foundation Graduate Research Fellowship Program under Grant No. DGE-1746045. 
This research was co-funded by the European Union (ERC, EACTP, 101142020), the Israel Science Foundation (grant number 514/20) and the Len Blavatnik and the Blavatnik Family Foundation. 
PW is supported by ERC Consolidator Grant VerNisQDevS (101086733).

Any views, opinions, findings, and conclusions or recommendations expressed in this material are those of the author(s) and do not necessarily reflect the views of the National Science Foundation, European Union, or the European Research Council Executive Agency. Neither the European Union nor any granting authority can be held responsible for them.

\printbibliography
% \bibliography{main.bib}

\clearpage
\newpage

\appendix

\clearpage
\newpage
\section{Characterizing $\QMA(2)$ as a purity test}
\label{sec:qma2_purity}
The proof of $\QMAP(2) = \NEXP$ uses the two unentangled proofs \emph{only} to implement a purity test. We suggest this is not coincidence: many works~\cite{harrow2013testing,bks17,Yu_2021,Yu_2022} take advantage of the fact that every protocol over separable states is equivalently a protocol over states with a pure subsystem.
We formalize this equivalence in a self-contained way.

The optimal success probability of a $\QMA(2)$ verifier $V$ is $h_{sep}(V) \defeq \max_{\rho \in \sep(X,Y)} \Tr[V \rho]$ \cite{harrow2013testing}.
By convexity, we only need to optimize over pure product states $\rho = \ketbra{\psi_X} \otimes \ketbra{\psi_Y}$. In fact, we only need to optimize over $\rho \in \cd(X,Y)$ such that $\Tr_Y[\rho]$ is pure, since these states contain all pure product states:
\begin{claim}
\label{fact:pure_rdm_means_product}
    Consider a state $\rho \in \cd(X,Y)$ where $\Tr_Y[\rho]$ is pure. Then $\rho$ is a product state.
\end{claim}
\begin{proof}
Diagonalize $\rho = \sum_i p_i \ketbra{\phi_i}$. Since $\Tr_Y[\rho] = \sum_i p_i \Tr_Y[\ketbra{\phi_i}]$ is some pure state $\ketbra{\psi}$, all $\Tr_Y[\ketbra{\phi_i}]$ must be proportional to $\ketbra{\psi}$. A pure state has a pure subsystem if and only if it is separable, so $\ket{\phi_i} = \ket{\psi}\ket{\chi_i}$ for all $i$. So $\rho = \ketbra{\psi} \otimes \sum_i p_i \ketbra{\chi_i}$ is a product state.
\end{proof}
\begin{fact}
    Consider a state $\rho = \ketbra{\psi} \otimes \ketbra{\phi}$. Then its reduced density matrix $\ketbra{\phi}$ is pure.
\end{fact}
So $h_{sep}(V)$ is equivalently an optimization over all states $\rho \in \cd(X,Y)$ where $\Tr_Y[\rho]$ is pure. In other words, $h_{sep}(V)$ is exactly
\begin{align*}
    &\max_{\rho \in \cd(X,Y)} \Tr[V \rho]
    \\
    \textnormal{such that }\quad  &rank(\Tr_{Y}[\rho]) = 1\,.
\end{align*}
Thus, the optimal acceptance probability of a $\QMA(2)$ protocol can be written as an SDP with a rank-1 constraint (i.e. pure \emph{subsystem}).

At the same time, any SDP with a rank-1 constraint can be converted to an SDP over separable states. For example, we may rewrite the above optimization as 
\begin{align}
    &\max_{\rho \in \sep(XY, Y')} \Tr[V \rho] \nonumber
    \\
     \textnormal{such that }\quad  &\Tr_{X}[\rho] \in Sym(Y, Y')\,.
     \label{eqn:constraint}
\end{align}
\cref{eqn:constraint} specifies the SWAP test on $(Y, Y')$ to accept with probability $1$. 
This is because optimizing over separable states is equivalent to optimizing over product states, and for product states, the SWAP test passes with probability $1$ iff the states are rank-1 (i.e. pure).

One can incorporate \cref{eqn:constraint} into a new verifier $V'$, which with probability $p$ applies the SWAP test on $(Y, Y')$, and otherwise applies $V$. Then $\max_{\rho \in \sep(XY, Y')} \Tr[V' \rho]$ is affine to $h_{sep}(V)$, with precision depending on $p$. This is essentially the content of the product test of \cite{harrow2013testing}.

We also describe how $\QMA(2)$ essentially applies ``purity testing'' to extendable states.
Extendable states can be thought of as ``weakly'' separable.
\begin{definition}[$k$-extendable bipartite state]
    Fix $k > 1$. A state $\rho_{AB} \in \cd(A,B)$ is $k$-extendable with respect to $B$ if there exists a state $\sigma \in \cd(A,B_1, \dots, B_k)$ that is invariant on permutation of the subsystems $\{B_1, \dots, B_k\}$ and satisfies $\Tr_{B_2, \dots, B_k}[\sigma] = \rho_{AB}$.
\end{definition}
\begin{fact}
\label{fact:separable_means_extendable}
    Every separable state is $k$-extendable for all $k > 1$.
\end{fact}
\begin{proof}
    Any separable state $\rho = \sum_i p_i \rho_{A,i} \otimes \rho_{B,i}$ can be extended to $\sigma =\sum_i p_i \rho_{A,i} \otimes \rho_{B,i}^{\otimes k}$.
\end{proof}
In fact, a state is separable \emph{if and only if} it is $k$-extendable for all $k$ \cite{Doherty_2004}.

For \emph{pure} states, extendability is equivalent to separability. As a result, any optimization over pure separable states is equivalently an optimization over pure $k$-extendable states for any $k > 1$.
\begin{fact}
\label{fact:pure_extendable_is_separable}
    Any \emph{pure} bipartite state $\ket{\psi} \in \cd(A,B)$ is 2-extendable with respect to $B$ if and only if it is separable.
\end{fact}
\begin{proof}
The backwards direction follows from \Cref{fact:separable_means_extendable}, so we focus on the forwards direction. Consider the 2-extension $\sigma \in \cd(A,B_1,B_2)$, which can be diagonalized as $\sigma = \sum_i p_i \ketbra{\phi_i}$. Note that $\Tr_{B_2}[\sigma] = \ketbra{\psi}$ is pure; thus, $\sigma = \ketbra{\psi} \otimes \sum_i p_i \ketbra{\chi_i}$ by \Cref{fact:pure_rdm_means_product}. But since $\sigma$ is symmetric over $(B_1, B_2)$, $\Tr_{B_1}[\sigma] = \ketbra{\psi}$ is \emph{also} pure. This implies that $\Tr_{B_1}[\ketbra{\psi}] \otimes \sum_i p_i \ketbra{\chi_i}$ must be pure, and thus $\ket{\psi}$ must be separable.
\end{proof}

By contrast, optimization over \emph{mixed} extendable states is equivalent to that over arbitrary states:
\begin{claim}
An optimization problem over arbitrary states can be converted to one over $\poly(n)$-extendable states, and vice versa.
\end{claim}
\begin{proof}
    Consider the optimization $\max_{\rho \in \cd(A)} \Tr[V \rho]$. This is trivially equivalent to optimizing $\Tr[(V_A \otimes \mathbb{I}_B) \rho]$ over states $\rho \in \cd(A,B)$ that are $k$-extendable with respect to $B$.
    
    Fix any $k = \poly(n)$. Consider optimizing $\Tr[W \rho]$ over states $\rho \in \cd(A,B)$ that are $k$-extendable with respect to $B$. This is equivalent to the following optimization:
    \begin{align*}
    &\max_{\rho \in \sep(A, B_1, \dots, B_k)} \Tr[(W_{A,B_1} \otimes \mathbb{I}_{B_2,\dots,B_k}) \rho] \nonumber
    \\
     \textnormal{such that }\quad  &\Tr_{A}[\rho] \in Sym(B_1,\dots,B_k)\,.        
    \end{align*}
    Moreover, one can implement the constraint by projecting $\rho$ to the symmetric subspace, and accepting only if the projection succeeds.
\end{proof}
From this perspective, the power of unentanglement in $\QMA(2)$ is exactly the ability to test purity of extendable states. See \Cref{table:pure_extendable} for a summary of these claims.

\begin{table}[h]
    \centering
    \begin{tabular}{|c|c|c|c|}
    \hline
        \emph{optimization over} & arbitrary states & $\poly(n)$-extendable states & separable states \\
             \hline
        mixed states allowed & $\QMA$ &  $\QMA$ &  $\QMA(2)$\\
    \hline
        pure states only &  $\QMA$ &   $\QMA(2)$ &  $\QMA(2)$ \\
        \hline
    \end{tabular}
    \caption{\footnotesize Categorizing quantum verification protocols by the purity and extendability of the quantum proof.}
    \label{table:pure_extendable}
\end{table}

\clearpage
\newpage

\section{On gap amplification of $\QMAP$ and $\QMAPP$} 
\label{sec:gapamplification}

Both $\QMAP$ and $\QMAPP$ can trivially simulate any $\QMA$ protocol by using only part $A$ of the quantum proof, i.e. ignoring the proof's separable subsystem.
But at large enough completeness-soundness gap, both $\QMAP$ and  $\QMAPP$ can be simulated in $\QMA$. This relies on the following claim:
\begin{claim}
    \label{claim:approxlemma}
    Consider any tripartite state $\ket{\psi}$ supported on registers $A,B,C$. Then there exists a state $\ket{\chi} \in \cpp$ such that that $|\bra{\chi}\ket{\psi}|^2 \ge \frac{1}{\dim(C)}$.
\end{claim}
\begin{proof}
Schmidt decompose $\ket{\psi}$ on the $AB|C$ bipartition as $\ket{\psi} = \sum_i \alpha_i \ket{i}_{AB}\ket{\phi_i}_{C}$. By pigeonhole principle, there is an $i^\star$ such that $|\alpha_{i^\star}|^2$ is at least $1$ divided by the number of Schmidt vectors. The number of Schmidt vectors is at most $\min(\dim(B), \dim(C)) \le \dim(C)$, so $|\alpha_{i^\star}|^2 \ge \frac{1}{\dim(C)}$. Then the state $\ket{\chi} \defeq \ket{i^\star}_{AB}\ket{\phi_{i^\star}}_{C}$ has at least $\frac{1}{\dim(C)}$ squared overlap with $\ket{\psi}$.
\end{proof}

\begin{claim}
\label{claim:qmapurif_in_qma}
    There exist constants $c,s$ such that $\QMAP^{c,s} \subseteq \QMA$ and $\QMAPP^{c,s} \subseteq \QMA$.
\end{claim}

\begin{proof}
This proof follows \cite{sqma,bassirian2023qmaplus}. We focus on $\QMAP$; the proof for $\QMAPP$ is identical. Consider any $\QMAP^{c,s}$ protocol with verifier $V$. We directly run $V$ in $\QMA$.

By \Cref{claim:approxlemma}, any quantum proof can be written as $\ket{\psi} = \sqrt{\ell} \ket{g} + (\sqrt{1-\ell}) \ket{b}$, where $\ket{g}$ is a unit vector in  $\cpp$, 
$\ket{b}$ is some other unit vector, $0 < \ell \le 1$ is a positive uniform constant,
and $\braket{b}{g} = 0$.
Recall also that $\ccc \subseteq \cpp \subseteq \cm$.
\begin{itemize}
    \item The completeness proof is unchanged, so the simulator accepts with probability at least $c$.
    \item In soundness, the simulator accepts with probability at most
    \begin{align*}
    \bra{\psi} V \ket{\psi} &= \ell \bra{g} V\ket{g} + (1-\ell)\bra{b} V\ket{b} + \sqrt{\ell(1-\ell)} \cdot (\bra{b}V\ket{g} + \bra{g}V\ket{b})
        \\
        &\le \ell \bra{g}V\ket{g} + (1-\ell) \bra{b}V\ket{b} + 2 \sqrt{\ell(1-\ell)} \cdot \sqrt{\bra{b}V\ket{b} \cdot \bra{g}V\ket{g}}
        \\
        &\le \ell \cdot s + (1-\ell) \cdot 1 + 2 \sqrt{\ell(1-\ell)} \cdot \sqrt{s}\,.
    \end{align*}
    The first inequality follows by Cauchy-Schwarz since $V$ is PSD (i.e. $V = L L^\dagger$).
\end{itemize}
For any positive constant $\ell$, there exist constant $c,s$ where $1-\ell + \ell s + 2 \sqrt{s \cdot \ell(1-\ell)} < c$,\footnote{For example, $c = 1 - \frac{\ell^2}{4}$ and $s = \frac{\ell}{4}$ is sufficient.} and so the protocol successfully distinguishes completeness and soundness.
\end{proof}
\begin{remark}
\label{remark:qmapurif2_in_qma2}
    As in \cite{sqma,bassirian2023qmaplus}, \Cref{claim:qmapurif_in_qma} naturally extends to $\QMA(k)$. For any constant $k$, there exist constants $c,s$ such that $\QMAP^{c,s}(k) \subseteq \QMA(k)$.
\end{remark}

We also prove that $\QMAPP = \NEXP$ in \Cref{sec:qmaw_equals_nexp}. As a result, $\QMAPP$ does not have gap amplification unless $\QMA = \NEXP$. Note the same is true for the class $\QMA^+$~\cite{jeronimo2023power,bassirian2023qmaplus}.

Nonetheless, $\QMAP$ may exhibit gap amplification. 
The naive strategy of parallel repetition fails, since a partial measurement can destroy the guarantee of internal separability. But it is possible that other sophisticated methods (e.g.~\cite{marriott2005quantumarthurmerlingames,harrow2013testing}) could be used to prove such a statement.

\clearpage
\newpage
\section{A complete problem for $\QMA$ with a restricted set of proofs}
\label{sec:completeproblem}
Inspired by the work of \cite{Hayden_2013,Gutoski_2015}, we define a complete problem for $\QMA_\cw$. It asks whether an isometry\footnote{An isometry is a map that preserves inner products, but unlike a unitary, may not be surjective.} has an output with good fidelity to a state in $\cw$ that has been padded with extra qubits:\footnote{The works of \cite{Hayden_2013,Gutoski_2015} consider isometry output problems without padded qubits. However, their proofs require a state injection gadget that relies on special properties of $\cw$ and gap amplification of $\QMA_{\cw}$.}

\newcommand{\paddedisometry}{$(\alpha,\beta)$ \textsc{Padded-}$\cw$ \textsc{Isometry Output}}
\begin{problem}[\paddedisometry{}]\phantom{for newline}
\label{padded_problem}
    \begin{itemize}
        \item[] \textbf{Input:} Description of a quantum circuit $C$ implementing an isometry $U$ from $X\defeq \CC^{2^n}$ to $A\otimes B\defeq \CC^{2^{m_1}}\otimes\CC^{2^{m_2}}$; i.e. $U: X \to A \otimes B$ acts as $U(\rho) \defeq C(\rho\otimes\ketbra{0}^{\otimes (m_1 + m_2 - n)})C^\dagger$.
        \item[] \textbf{YES case:} There exists an $n$-qubit state $\rho\in \cd(X)$ and $m_1$-qubit state $\sigma \in \cw$ for which $U(\rho)$ has fidelity\footnote{Following \cite{Hayden_2013,Gutoski_2015}, we use the ``squared'' version of fidelity, i.e. $F(\rho, \sigma) \defeq \left(\Tr[\sqrt{\sqrt{\rho} \sigma \sqrt{\rho}}] \right)^2$.} at least $\alpha$ with $\sigma \otimes \ketbra{0}^{\otimes m_2}$; i.e.
         $F(U(\rho),\sigma\otimes\ketbra{0}^{\otimes m_2})\geq\alpha$.
         
        \item[] \textbf{NO case:} For every $n$-qubit state $\rho\in \cd(X)$ and $m_1$-qubit state $\sigma \in \cw$, $U(\rho)$ has fidelity at most $\beta$ with $\sigma \otimes \ketbra{0}^{\otimes m_2}$; i.e.
        $F(U(\rho),\sigma \otimes\ketbra{0}^{\otimes m_2})\leq\beta$.
    \end{itemize}
\end{problem}
In this section, we prove \Cref{padded_problem} is always a complete problem for $\QMA_{\cw}$:
\begin{theorem}
\label{thm:completeproblem}
    \paddedisometry{} is complete for $\QMA_{\cw}^{\alpha,\beta}$.
\end{theorem}

We require some basic properties about fidelity.

\begin{fact}[{e.g. \cite[Chapter 9]{Nielsen_Chuang_2010}}]\label{fact:fidelity}
The following holds for every $\rho, \sigma \in \cd(X)$:
\begin{enumerate}[label=(\roman*)]
    \item\label{enu:unitary-inv} For any unitary $U$,
        $F(U\rho U^\dagger, U\sigma U^\dagger) = F(\rho, \sigma).$
    \item\label{enu:data-processing} For any CPTP map $\Phi$,
    $F(\Phi(\rho), \Phi(\sigma)) \ge F(\rho, \sigma).$
    \item\label{enu:pure-overlap} If $\rho$ is pure, then
            $F(\rho, \sigma) = \langle \rho |\sigma |\rho \rangle.$
\end{enumerate}
\end{fact}
We first show that this problem can be decided in $\QMA_{\cw}$:

\begin{claim}\label{claim:isometry-problem-inclusion}
    For all $\beta\leq s < c \leq \alpha$, \paddedisometry{} is in $\QMA_\cw^{c,s}$.
\end{claim}
\begin{proof} 
We describe the $\QMA_\cw^{c,s}$ protocol. Merlin sends an $m_1$-qubit proof $\sigma \in \cw$. Arthur appends $\ket{0}^{\otimes m_2}$ to the proof, applies the inverse of the quantum circuit $C$, and measures the last $m_1 + m_2 - n$ qubits in the computational basis. Arthur accepts if all measured values are $0$.

We now analyze completeness. Let $C$ be the circuit implementing the isometry $U$, and let $\rho \in \cd(X), \sigma \in \cw$ be such that $F(U(\rho), \sigma \otimes \ketbra{0}^{\otimes m_2}) \ge \alpha$. 
Using \Cref{enu:unitary-inv} of \Cref{fact:fidelity},
\begin{align*}
    \alpha &\le F(U(\rho), \sigma \otimes \ketbra{0}^{\otimes m_2})
    \\
    &= F(C(\rho \otimes \ketbra{0}^{\otimes (m_1 + m_2 - n)})C^\dagger, \sigma \otimes \ketbra{0}^{\otimes m_2})
    \\
    &= F(\rho \otimes \ketbra{0}^{\otimes (m_1 + m_2 - n)}, C^\dagger(\sigma \otimes \ketbra{0}^{\otimes m_2})C)\,.
\end{align*}
Imagine tracing out the first $n$ qubits of the states $\rho \otimes \ketbra{0}^{\otimes (m_1 + m_2 - n)}$ and $C^\dagger(\sigma \otimes \ketbra{0}^{\otimes m_2})C$. 
By \Cref{enu:data-processing} of \Cref{fact:fidelity},
these reduced states have fidelity at least $\alpha$. Since the first reduced state is pure ($\ket{0}^{\otimes (m_1 + m_2 - n)}$), \Cref{enu:pure-overlap} of \Cref{fact:fidelity} implies that the fidelity is equal to the probability that measuring the second reduced state in the computational basis gives all $0$'s. This value is the acceptance probability of the protocol, and thus at least $\alpha \ge c$.

Now we analyze soundness. By the same argument, for any $\rho \in \cd(X), \sigma \in \cw$, we have
\begin{align*}
    \beta \ge  F(\rho \otimes \ketbra{0}^{\otimes (m_1 + m_2 - n)}, C^\dagger(\sigma \otimes \ketbra{0}^{\otimes m_2})C)\,.
\end{align*}
From here, we use a variant of Uhlmann's theorem (e.g. \cite[Exercise 9.15]{Nielsen_Chuang_2010}): the maximum fidelity of any two states is achieved by \emph{partial} purifications of the states. Precisely, for any two states $\rho, \eta \in \cd(A)$ and a partial purification of one state ($\rho' \in \cd(A,B)$ where $\rho = \Tr_B[\rho']$), we have  $F(\rho,\eta)=\max_{\eta' \in \cd(A,B); \Tr_B[\eta']=\eta}F(\rho',\eta')$.

Consider tracing out the first $n$ qubits of $C^\dagger(\sigma \otimes \ketbra{0}^{\otimes m_2})C$. Then its fidelity with $\ketbra{0}^{\otimes (m_1 + m_2 - n)}$ is achieved by the maximum fidelity of  $C^\dagger(\sigma \otimes \ketbra{0}^{\otimes m_2})C$ with purifications of $\ketbra{0}^{\otimes (m_1 + m_2 - n)}$, i.e. states of the form $\rho \otimes \ketbra{0}^{\otimes (m_1 + m_2 - n)}$. So, the fidelity of these reduced states is at most $\beta$.

Again, since $\ketbra{0}^{\otimes (m_1 + m_2 - n)}$ is a pure state, \Cref{enu:pure-overlap} of \Cref{fact:fidelity}  implies that the fidelity is equal to the probability measuring the other state gives all $0$'s. This value is Arthur's acceptance probability, and thus at most $\beta \le s$.
\end{proof}

We now show that this problem is hard for $\QMA_{\cw}$:
\begin{claim}\label{claim:general-hard-problem}
    For all $s \le \beta < \alpha\leq c$,  \paddedisometry{} is $\QMA_\cw^{c,s}$-hard.
\end{claim}
\begin{proof}
Let $L=(L_{Yes},L_{No})$ be a (decision) promise problem in $\QMA_\cw^{c,s}$. Then for every input $x\in L_{Yes}\cup L_{No}$, there is a quantum circuit $C_x$ uniformly computable in polynomial time in $|x|$ that represents Arthur's verification protocol. This circuit $C_x: M \otimes S \to D \otimes G$ maps Merlin's quantum proof ($M$) and scratch qubits ($S$) initialized at $\ket{0}$, to a decision qubit ($D$) and ``garbage'' output ($G$); see \Cref{fig:hardness-original-circ} for an illustration of the circuit.

We uniformly encode $C_x$ as an instance of \paddedisometry{} when $\alpha \le c$ and $\beta \ge s$. Let $\tilde{C}_x \defeq C_x^\dagger \circ (X_D \otimes \mathbb{I}_G)$ be the unitary implementing an isometry $U: G \to W \otimes S$, where $X_D$ is a Pauli-$X$ gate applied to register $D$. See \Cref{fig:hardness-reduction-circ} for an illustration of the isometry.

    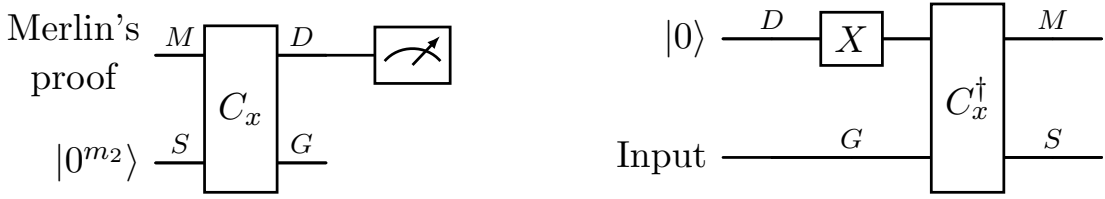
\begin{figure}[t]
        \centering
        \begin{subfigure}[t]{0.45\textwidth}
            \centering
            \begin{tikzpicture}
            \node[scale=1.3] {
            \begin{quantikz}
            \lstick{Merlin's\\proof} &\gate[2]{C_x}\wire[l][1]["M"{above}]{q}&\wire[l][1]["D"{above}]{q}&\meter{}\\
                \lstick{$\ket{0^{m_2}}$}&\wire[l][1]["S"{above}]{q}&\wire[l][1]["G"{above}]{q}
        \end{quantikz} };
    \end{tikzpicture}
            \caption{\footnotesize A verification circuit for a problem in $\QMA_{\cw}$, given the input $x$.}
            \label{fig:hardness-original-circ}
        \end{subfigure}
        \hspace{2em}
        \begin{subfigure}[t]{0.45\textwidth}
            \centering
            \begin{tikzpicture}
            \node[scale=1.3] {
\begin{quantikz}
    \lstick{$\ket{0}$} &&\gate{X}\wire[l][2]["D"{above}]{q}& \gate[2]{C_x^\dagger} \wire[l][1][
    ]{q} &&\wire[l][2]["M"{above}]{q}\\
    \lstick{Input} &&& 
    \wire[l][2]["G"{above,pos=0.5}]{q}&& \wire[l][2]["S"{above,pos=0.5}]{q}
        \end{quantikz} };
    \end{tikzpicture}
            \caption{
            \footnotesize
            The verification circuit recast as an instance of \paddedisometry{}. Note that the Pauli-$X$ gate flips $\ket{0}$ to $\ket{1}$ and vice versa.}
            \label{fig:hardness-reduction-circ}
        \end{subfigure}
        \caption{\footnotesize An illustration of the reduction in \Cref{claim:general-hard-problem}.}
        \label{fig:hardness-reduction}
    \end{figure}

We now verify the correctness of the reduction.
The acceptance probability of $C_x$ is the chance that measuring the decision qubit in the computational basis outputs $1$, which is equal to 
\begin{align*}
    \max_{\sigma \in \cw} F\left(\ketbra{1}, \Tr_G[C_x(\sigma \otimes \ketbra{0}^{\otimes \log |S|})C_x^\dagger]\right)\,.
\end{align*}
By Uhlmann's theorem, the fidelity is achieved by maximizing over partial purifications, so the acceptance probability is equal to
\begin{align*}
     &\max_{\sigma \in \cw} \max_{\rho \in \cd(G)} F\left(\ketbra{1} \otimes \rho, C_x(\sigma \otimes \ketbra{0}^{\otimes \log |S|})C_x^\dagger\right)
     \\
     = 
     &\max_{\sigma \in \cw} \max_{\rho \in \cd(G)} F\left(C_x^\dagger (\ketbra{1} \otimes \rho)C_x, \sigma \otimes \ketbra{0}^{\otimes \log |S|}\right)
     \\
     =
     &\max_{\sigma \in \cw} \max_{\rho \in \cd(G)} F\left(\tilde{C}_x (\ketbra{0} \otimes \rho)\tilde{C}_x^\dagger, \sigma \otimes \ketbra{0}^{\otimes \log |S|}\right)
     \,.
\end{align*}
This is exactly the quantity described in \paddedisometry{}.

Suppose $x \in L_{Yes}$. Then there is a state $\sigma \in \cw$ that Arthur accepts with probability at least $c$. By the above equality, there is some $\rho \in \cd(G)$ such that $U(\rho)$ has fidelity at least $c$ with $\sigma \otimes \ketbra{0}^{\otimes \log |S|}$. So this is a YES instance of \paddedisometry{} for $\alpha \le c$.

Now suppose $x \in L_{No}$. Then for all states $\sigma \in \cw$, Arthur accepts with probability at most $s$. 
By the above equality, for all $\rho \in \cd(G)$ and $\sigma \in \cw$, the state $U(\rho)$ has fidelity at most $s$ with $\sigma \otimes \ketbra{0}^{\otimes \log |S|}$. So this is a NO instance of \paddedisometry{} for $\beta \ge s$. 
\end{proof}

\Cref{thm:completeproblem} follows from \Cref{claim:isometry-problem-inclusion} and  \Cref{claim:general-hard-problem}.

\subsection{Removing the extra qubits}
What happens if we remove the extra qubits from  \Cref{padded_problem}? We show that this modified problem is still complete for $\QMA_{\cw}$ under some technical assumptions.
\newcommand{\isometryprob}{$(\alpha,\beta)$ $\cw$ \textsc{Isometry Output}}
\begin{problem}[\isometryprob{}]\phantom{for newline}
\label{problem_nopad}
    \begin{itemize}
        \item[] \textbf{Input:} Description of a quantum circuit $C$ implementing an isometry $U$ from $X\defeq \CC^{2^n}$ to $A\defeq \CC^{2^{m}}$; i.e. $U: X \to A$ acts as $U(\rho) \defeq C(\rho\otimes\ketbra{0}^{\otimes (m - n)})C^\dagger$.
        \item[] \textbf{YES case:} There exists an $n$-qubit state $\rho\in \cd(X)$ and $m$-qubit state $\sigma \in \cw$ for which $U(\rho)$ has fidelity at least $\alpha$ with $\sigma$; i.e.
         $F(U(\rho),\sigma)\geq\alpha$.
         
        \item[] \textbf{NO case:} For every $n$-qubit state $\rho\in \cd(X)$ and $m$-qubit state $\sigma \in \cw$, $U(\rho)$ has fidelity at most $\beta$ with $\sigma$; i.e.
        $F(U(\rho),\sigma)\leq\beta$.
    \end{itemize}
\end{problem}
Since \Cref{problem_nopad} is a special case of \Cref{padded_problem}, it is also in $\QMA_{\cw}$.
\begin{fact}
    For all $\beta\leq s < c \leq \alpha$, \isometryprob{} is in $\QMA_\cw^{c,s}$.
\end{fact}

However, we can only show \Cref{problem_nopad} is $\QMA_{\cw}$-hard when $\cw$ satisfies some conditions. 
The proof in \Cref{claim:general-hard-problem} creates an isometry by inverting an arbitrary verification circuit in $\QMA_{\cw}$.
But the initial workspace of a $\QMA_{\cw}$ verifier contains a state $\rho \in \cw$, \emph{and} a set of scratch qubits (which we can assume to be all $\ket{0}$).
In order to create an instance of \Cref{problem_nopad}, we must do two things: (1) in completeness, ``absorb'' the scratch qubits into a state in $\cw$, and (2) in soundness, somehow ``poison'' all outputs of the isometry to be far from all states in $\cw$.
The first depends on properties of $\cw$; to do the second, we slightly generalize a technique of \cite{Hayden_2013,Gutoski_2015}, which requires a third assumption on $\cw$. The completeness-soundness gap of the instance depends on the quality of the poisoning operation.

\begin{claim}\label{claim:general-hard-problemNoPad}
    Let $\cw$ be a set of states with the following properties:
    \begin{enumerate}
        \item (Absorbing) For any $n$-qubit $\rho \in \cw$ and polynomials $p_1,p_2$, $\ketbra{0}^{p_1(n)} \otimes \rho \otimes \ketbra{0}^{\otimes p_2(n)} \in \cw$.
        \item (Poisonable) There is a function $0 \le g(n) \le 1$ and an efficient, uniformly preparable family of isometries $\{P_i\}$ such that for any polynomial $q$ and $n+q(n)$-qubit state $\sigma$, $(P_n \otimes \mathbb{I}_{q(n)})(\sigma)$ has fidelity at most $g(n)$ with any state in $\cw$. 
        \item (Projecting) There is a polynomial $r$ such that for each $n$, $r(n) \ge n$, and with the $n$-qubit projector $\Pi_0$ to all $0$'s, if the $r(n)$-qubit state $\rho \in \cw$, then $(\mathbb{I} \otimes \Pi_0) \rho (\mathbb{I} \otimes \Pi_0^\dagger)$ (appropriately normalized) is in $\cw$.
    \end{enumerate} 
    Then for all $\alpha\leq c$ and $\beta\geq s+g+2\sqrt{sg}$, \isometryprob{} is $\QMA_\cw^{c,s}$-hard.
\end{claim}

\begin{proof}
We start as in the proof of \Cref{claim:general-hard-problem}.
Let $L=(L_{Yes},L_{No})\in\QMA_\cw^{c,s}$ be a (decision) promise problem in $\QMA_{\cw}^{c,s}$. Then for every input $x\in L_{Yes}\cup L_{No}$, there is a quantum circuit $C_x$ uniformly computable in polynomial time in $|x|$ that represents Arthur's verification protocol.
This circuit $C_x: M \otimes S \to D \otimes G$ maps Merlin's quantum proof ($M$) and scratch qubits ($S$) initialized at $\ket{0}$, to a decision qubit ($D$) and ``garbage'' output ($G$); see \Cref{fig:hardness-original-circ} for an illustration of the circuit.

We now construct the isometry $U$. 
We start with the construction in the proof \Cref{claim:general-hard-problem}, and ``poison'' the output when the scratch qubits are incorrectly set.
Suppose $n$ is the number of qubits in Merlin's proof.
The unitary that implements $U$ maps $Z \otimes D \otimes G \to Z \otimes M \otimes S$, where register $Z$ is used to implement the
poisoning operation $P_n$.
The isometry is defined through the unitary
$Y \circ (\mathbb{I}_Z \otimes C_x^\dagger) 
\circ (\mathbb{I}_Z \otimes X_D \otimes \mathbb{I}_{G})$, where $Y$ applies $P_n$ controlled on the $S$ register being all $0$'s. 
The isometry sets the input qubits of registers $Z$ and $D$ all to $0$.
See \Cref{fig:hardnessNoPad-reduction} for an illustration.

We now verify the correctness of the reduction.
As in the proof of \Cref{claim:general-hard-problem}, the acceptance probability of $C_x$ is equal to 
\begin{align*}
    &\max_{\sigma \in \cw} F\left(\ketbra{1}, \Tr_G[C_x(\sigma \otimes \ketbra{0}^{\otimes \log |S|})C_x^\dagger]\right)\,
    \\
    =
     &\max_{\sigma \in \cw} \max_{\rho \in \cd(G)} F\left(\ketbra{1} \otimes \rho, C_x(\sigma \otimes \ketbra{0}^{\otimes \log |S|})C_x^\dagger\right)
     \\
     = 
     &\max_{\sigma \in \cw} \max_{\rho \in \cd(G)} F\left(C_x^\dagger (\ketbra{1} \otimes \rho)C_x, \sigma \otimes \ketbra{0}^{\otimes \log |S|}\right)\,.
\end{align*}
We now add the register $Z$, and apply the controlled unitary $Y$ to both states. This does nothing to the second state since it has all $0$'s in the $S$ register. So the acceptance probability of $C_x$ equals 
\begin{align*}
         \max_{\sigma \in \cw} \max_{\rho \in \cd(G)} F\left( U(\rho), \ketbra{0}^{\log |Z| } \otimes \sigma \otimes \ketbra{0}^{\otimes \log |S|}\right)\,.
\end{align*}
Suppose $x \in L_{Yes}$. Then there exists a state $\rho \in \cd(G)$ and a quantum proof $\sigma \in \cw$ such that the fidelity $F(U(\rho), \ketbra{0}^{\log |Z| } \otimes \sigma \otimes \ketbra{0}^{\otimes \log |S|}) \ge c$. By the absorbing property of $\cw$, $\ketbra{0}^{\log |Z| } \otimes \sigma \otimes \ketbra{0}^{\otimes \log |S|} \in \cw$. So there is a quantum proof $\sigma' \in \cw$ such that $F(U(\rho), \sigma') \ge c$, which is a YES instance of \Cref{problem_nopad} whenever $c \ge \alpha$.

Suppose $x \in L_{No}$. Then for all states $\rho \in \cd(G)$ and quantum proofs $\sigma \in \cw$, the fidelity $F(U(\rho), \ketbra{0}^{\log |Z| } \otimes \sigma \otimes \ketbra{0}^{\otimes \log |S|}) \le s$.
Now we use the following consequence of Uhlmann's theorem:
\begin{fact}
\label{fact:fidelity_trick}
    Suppose $\ket{\psi}$ purifies a state $\mu$ via register $B$. Then for any projector $\Pi_0$,
\begin{align*}
\sqrt{F(\mu,\nu)}=\max_{ \Tr_B[\ketbra{\zeta}]=\nu}|\braket{\psi}{\zeta}|
&\leq\max_{\Tr_B[\ketbra{\zeta}]=\nu}|\bra{\psi}\Pi_0\ket{\zeta}|+\max_{\Tr_B[\ketbra{\zeta}]=\nu}|\bra{\psi}(\mathbb{I}-\Pi_0)\ket{\zeta}|\,
\\
&\le \sqrt{F\left(\mu, \frac{\Pi_0 \nu \Pi_0^\dagger}{\Tr[\Pi_0 \nu \Pi_0^\dagger]}\right)} +  \sqrt{F\left(\frac{(\mathbb{I} - \Pi_0) \mu (\mathbb{I} - \Pi_0)^\dagger}{\Tr[(\mathbb{I} - \Pi_0) \mu (\mathbb{I} - \Pi_0)^\dagger]}, \nu\right)}
\,.
\end{align*}
\end{fact}
In \Cref{fact:fidelity_trick}, we set $\mu \defeq U(\rho)$ and let $\nu$ be the witness in $\cw$ that maximizes fidelity; then $F(\mu, \nu)$ is the quantity described in \Cref{problem_nopad}. Let $\Pi_0$ select the subspace where the qubits in register $S$ are all set to $0$. 
We analyze the two terms above separately.
\begin{itemize}
\item Because $\nu \in \cw$ and $\cw$ is projecting, $\frac{\Pi_0 \nu \Pi_0^\dagger}{\Tr[\Pi_0 \nu \Pi_0^\dagger]}\in \cw$. Since $\cw$ is absorbing, this value is at most $\max_{\sigma \in \cw} \sqrt{F(U(\rho), \sigma \otimes \ketbra{0}^{\otimes \log |S|}) }$. Since $U(\rho)$ has all $0$'s in register $Z$ when it has all $0$'s in register $S$, this is at most $\max_{\sigma' \in \cw} \sqrt{F(U(\rho), \ketbra{0}^{\otimes \log |Z|} \otimes \sigma' \otimes \ketbra{0}^{\otimes \log |S|})} \le \sqrt{s}$.
\item By the definition of $\mu$ and $\Pi_0$, the projector $\mathbb{I} - \Pi_0$ selects the part of $U(\rho)$ where the controlled poisoning operation $Y$ had an effect. Because $\cw$ is poisonable, the state $(\mathbb{I} - \Pi_0) \mu (\mathbb{I} - \Pi_0)^\dagger$ (appropriately normalized) has fidelity at most $g$ with $\nu$. So this term is at most $\sqrt{g}$.
\end{itemize}
This is a NO instance of \isometryprob{} when $\beta \le (\sqrt{s} + \sqrt{g})^2 = s + g + 2 \sqrt{sg}$.
\end{proof}
    \begin{figure}
        \centering           
        \begin{tikzpicture}
            \node[scale=1.3] {
\begin{quantikz}
    \lstick[2]{S}&\wireoverride{n}\lstick{$\ket{0^\ell}$} && \wire[l][1]["Z"{above,pos=0.8}]{q} & \gate[2]{P_n}  &\wire[l][1]["Z"{above,pos=0.1}]{q}&\rstick[3]{A}\\
    &\wireoverride{n}\lstick{$\ket{0}$} &\gate{X}\wire[l][1]["D"{above}]{q}& \gate[2]{C_x^\dagger} \wire[l][1]["D"{above}]{q} &\wire[l][1]["M"{above}]{q} &\wire[l][1]["M"{above,pos=0.1}]{q}&\\
    \lstick{X}&\wireoverride{n}\midstick{Input} && 
    \wire[l][1]["G"{above,pos=0.8}]{q}& \ctrl{-1} \wire[l][1]["S"{above}]{q} &\wire[l][1]["S"{above,pos=0.1}]{q}&
\end{quantikz} };
\end{tikzpicture}
        \caption{\footnotesize 
        The verification circuit recast as an instance of \isometryprob{} as in \Cref{claim:general-hard-problemNoPad}.
         The poisoning circuit $P_n$ is controlled by register $S$, and applied whenever the qubits in $S$ are not all $0$'s.
        }
        \label{fig:hardnessNoPad-reduction}
    \end{figure}
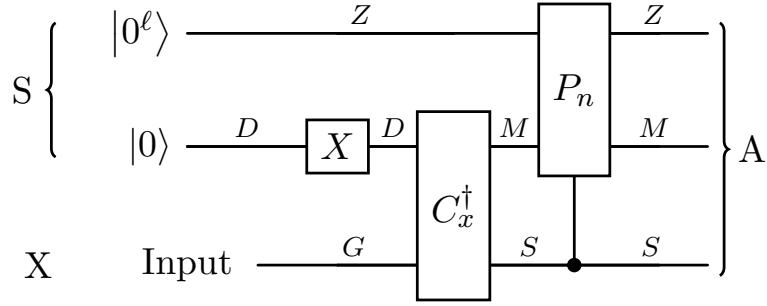

We apply \Cref{claim:general-hard-problemNoPad} to the sets  $\cm$ and $\cpp$ studied in this work. 
\begin{corollary}[A hard problem for $\QMAP,\QMAPP$]\label{cor:hard-for-pure-with-gap}
For any constant $\epsilon > 0$,
$\beta \le s + \epsilon$, and $\alpha \ge c$, $(\alpha,\beta)$ $\cm$ \textsc{Isometry Output} is $\QMAP^{c,s}$-hard, and $(\alpha,\beta)$ $\cpp$ \textsc{Isometry Output} is $\QMAPP^{c,s}$-hard.
\end{corollary}
\begin{proof}
We show $\cm$ and $\cpp$ each satisfy the three requirements in the hypothesis of \Cref{claim:general-hard-problem}.
\begin{enumerate}
    \item (Absorbing) For any state $\rho \in \cm$ or $\rho \in \cpp$,  extra qubits (for example set to 0) can always be absorbed to the $A$ or $B$ register. When the $A$ register is traced out, the state will remain separable.
    \item (Poisonable) For any constant $k$, consider any maximally entangled state over $2k$ qubits, for example $k$ EPR pairs. The poisoning operation prepares these EPR pairs on new qubits (i.e. register $Z$), and then applies SWAP gates to put one half of each EPR pair in register $B$ and the other half in register $C$. This creates nontrivial entanglement even after tracing out register $A$; by \cite[Equation 3.1]{Gutoski_2015}, this state has fidelity at most $2^{-k}$ with any separable state. By monotonicity of fidelity under partial trace, the poisoned state has fidelity at most $g \defeq 2^{-k}$ with any state in $\cm$ or in $\cpp$.
    \item (Projecting) Consider any $\rho \in \cm$. If we project $\Pi_0$ on some qubits in the $B$ register, then after tracing out the $A$ register, the reduced state is still separable. The same is true for $\cpp$.
\end{enumerate}
So the following holds for each $\cw  \in \{\cm, \cpp\}$: Given any $k \in \mathbb{Z}_{\ge 0}$, $\alpha\leq c$ and $\beta\geq s+2^{-k}+2\sqrt{s \cdot 2^{-k}}$, the problem \isometryprob{} is $\QMA_\cw^{c,s}$-hard. Choose $\epsilon \defeq 2^{-k} + 2\sqrt{s \cdot 2^{-k}}$, and the corollary follows.
\end{proof}
\Cref{cor:hard-for-pure-with-gap} implies that \isometryprob{} is complete for $\QMAP$ and $\QMAPP$ whenever there is a small amount of gap amplification. 
As of now, we only know this is true in specific situations. At large enough $(c,s)$, by \Cref{claim:qmapurif_in_qma} both classes are equal to $\QMA$. At small enough $(c,s)$, by \Cref{thm:qmapp-nexp}, $\QMAPP$ is equal to $\NEXP$. As a result, \isometryprob{} is complete for $\QMA$ and for $\NEXP$, depending on the promise gap:
\begin{corollary}
    There exist constants $(\alpha,\beta)$ for which both $(\alpha,\beta)$ $\cm$ \textsc{Isometry Output} and $(\alpha,\beta)$ $\cpp$ \textsc{Isometry Output} are $\QMA$-complete. Moreover, there exist other constants $(\alpha',\beta')$ for which $(\alpha',\beta')$ $\cpp$ \textsc{Isometry Output} is $\NEXP$-complete.
\end{corollary}
We may also interpret $\QMAP(2)$ as the class $\QMA_{\cm(2)}$, where $\cm(2)$ are all states $\rho_1 \otimes \rho_2$ such that $\rho_1, \rho_2 \in \cm$. The proof of \Cref{cor:hard-for-pure-with-gap} trivially extends to $\cm(2)$. Since at large enough $(c,s)$, $\QMAP(2)$ is equal to $\QMA(2)$, we get the following:
\begin{corollary}
    There exist constants $(\alpha,\beta)$ for which $(\alpha,\beta)$ $\cm(2)$ \textsc{Isometry Output} is $\QMA(2)$-complete. Moreover, there exist other constants $(\alpha',\beta')$ for which $(\alpha',\beta')$ $\cm(2)$ \textsc{Isometry Output} is $\NEXP$-complete.
\end{corollary}
\clearpage
\newpage

\clearpage
\newpage
\section{Other computational models where unentanglement is powerful}
\label{sec:more_unentanglement_power}
\subsection{Towards a new understanding of $\QMA^+$}
\Cref{lemma:swaptest_rank1restrict} clarifies recent work on  the complexity class $\QMA^+(k)$~\cite{jeronimo2023power,bassirian2023qmaplus}.
This class is defined as $\QMA^+(k) \defeq \QMA_{\cw_{\ket{+}}}(k)$, where $\cw_{\ket{+}} \defeq \{ \sum_i a_i \ket{i} \ |\  (\forall i)\ a_i \ge 0 \}$ is the set of quantum states with non-negative amplitudes in the computational basis.

Jeronimo and Wu found that $\QMA^+(2)$ at \emph{large} enough completeness-soundness gap can be simulated by $\QMA(2)$, but can decide $\NEXP$ at \emph{small} enough completeness-soundness gap~\cite{jeronimo2023power}. They conjecture that $\QMA(2) = \NEXP$, since this would follow if $\QMA^+(2)$ exhibits gap amplification. However, \cite{bassirian2023qmaplus} found a barrier to this approach, as gap amplification of $\QMA^+$ would imply the unlikely $\QMA = \NEXP$. One interpretation of this result is that the power of $\QMA^+$ comes from restricting quantum proofs to $\cw_{\ket{+}}$. This suggests the following question:
\begin{question}
\label{question:beyondqmaplus}
        Is there a set of quantum proofs $\cw$ where gap amplification of both $\QMA_{\cw}$ and $\QMA_{\cw}(2)$ implies $\QMA(2) = \NEXP$, and no other complexity collapse? 
\end{question}
We remark that \Cref{question:beyondqmaplus} is answered affirmatively by the set $\cm$.  Notably, a slight adjustment of $\QMA^+$ also affirmatively answers \Cref{question:beyondqmaplus}. Expand the set of quantum proofs to \emph{density matrices} with non-negative elements: $\cw_{+} \defeq \{ \rho \ |\ (\forall i,j) \ \rho_{ij} \ge 0 \}$. As before, $\QMA_{\cw_+}(k)$ at large enough completeness-soundness gap can be simulated by $\QMA(k)$. Moreover, $\QMA_{\cw_{+}} \subseteq \EXP$, as the  non-negative constraints can be added to an exponential-size semidefinite program. Finally, we can use \Cref{lemma:swaptest_rank1restrict} to show $\QMA_{\cw_{+}}(2) = \NEXP$:
\begin{claim}
    $\QMA_{\cw_{+}}(2)$ with a small constant completeness-soundness gap can decide $\NEXP$.
\end{claim}
\begin{proof}
First, note that any rank-1 density matrix with non-negative elements is in fact a non-negative amplitude state; thus, $\cw_{\ket{+}} = \purerestrict(\cw_+)$. By \Cref{lemma:swaptest_rank1restrict}, if $\cw_+$ has the continuity condition, then $\QMA^+ = \QMA_{\cw_{\ket{+}}} \subseteq \QMA_{\cw_+}(2)$. The result follows since $\QMA^+ = \NEXP$~\cite{bassirian2023qmaplus}.

We now verify that $\cw_+$ has the continuity condition. By a standard fact of linear algebra, the maximum eigenvector of any $\rho \in \cw_+$ can be chosen to have non-negative amplitudes, i.e. in $\cw_{\ket{+}} = \purerestrict(\cw_+)$. So then the continuity condition holds with $f(x) = x$. 
\end{proof}
\Cref{lemma:swaptest_rank1restrict} provides an interpretation for the result of \cite{bassirian2023qmaplus}. The complexity class $\QMA_{\cw_+}$ restricts quantum states as \emph{density matrices}. Here, we know that two unentangled proofs are more powerful than one: $\QMA_{\cw_+} \ne \QMA_{\cw_+}(2)$ unless $\EXP = \NEXP$. This directly stems from the SWAP test (i.e. purity test) in \Cref{lemma:swaptest_rank1restrict}. Thus, the power of $\QMA^+$ comes not from restricting quantum proofs to non-negative elements, but from \emph{further} restricting $\cw_+$ to \emph{pure states}. Because of this, $\QMA^+$ cannot have gap amplification (unless $\QMA = \NEXP$), but $\QMA_{\cw_+}$ plausibly can.

\subsection{Proper states}
\emph{Proper states} are states of the form $\ket{\psi} = \sum_{i=1}^d \pm \frac{1}{\sqrt{d}}\ket{i}$. 
It was shown in~\cite{power_of_unentanglement,beigi2009np} that optimizing over proper states of polynomial dimension is $\NP$-hard with a constant gap.
We show how this can be used to build a computational model where having one quantum proof is strictly weaker than having two quantum proofs, assuming $\P \ne \NP$. 

Consider the set of real density matrices with equal values on the diagonal; i.e.
\begin{align*}
    \wproper \defeq  \{ \rho \in \mathbb{R} \ |\ (\forall i,j) \ \rho_{ii} = \rho_{jj} \}\,.
\end{align*}
Consider a variant of $\QMA$ allowing only $O(\log n)$ space, and the quantum proof restricted to states in $\wproper$. Then this class can be decided in $\P$, since there is an efficient semidefinite program to calculate the maximum acceptance probability.  

On the other hand, note that the pure states of $\wproper$ are exactly the proper states of~\cite{power_of_unentanglement,beigi2009np}. By \Cref{lemma:swaptest_rank1restrict}, the associated variant of $\QMA(2)$ can simulate $\NP$, assuming that $\wproper$ satisfies the continuity condition. We verify this now:
\begin{claim}
    $\wproper$ has the continuity condition.
\end{claim}
\begin{proof}
    Consider any element $\rho \in \wproper$. Since $\rho$ is real, we can always choose the maximum eigenvector to be real-valued; let it be $\ket{\psi} \defeq \sum_i a_i \ket{i}$. Let the maximum eigenvalue is $\lambda_1$; then $\rho - \lambda_1 \ketbra{\psi}$ is a PSD matrix, and so all diagonal values are non-negative. Thus,
    \begin{align*}
        \sum_{i} |1/n - \lambda_1 a_i^2| = \sum_{i} 1/n - \lambda_1 a_i^2 = \Tr[\rho - \lambda_1 \ketbra{\psi}]  =  1-\lambda_1\,.
    \end{align*}
Let $X$ be the random variable with value $n \lambda_1 a_i^2$. Then $\mathbb{E}[X] = \lambda_1 \sum_i a_i^2 = \lambda_1$. Since $0 \le X \le 1$, $Var(X) \le \lambda_1(1-\lambda_1)$. By Chebyshev, $\Pr[X \le \lambda_1 - k \sqrt{\lambda_1(1-\lambda_1)}] \le \Pr[X \le \lambda_1 - k \sigma] \le \frac{1}{k^2}$.

We investigate the squared overlap of $\ket{\psi}$ with any proper state, which is maximized at value $\frac{1}{n} (\sum_i |a_i|)^2$. We lower-bound this quantity by ignoring the small terms from Chebyshev's inequality:
\begin{align*}
    \frac{1}{n} \left( n \cdot \left(1 - \frac{1}{k^2}\right) \cdot \sqrt{\frac{\lambda_1 - k\sqrt{\lambda_1(1-\lambda_1)}}{n \cdot \lambda_1}} \right)^2 = \left(1 - \frac{1}{k^2}\right)^2 \cdot \left(1 - k\sqrt{\frac{1-\lambda_1}{\lambda_1}}\right)\,.
\end{align*}
Let $c = \frac{1-\lambda_1}{\lambda_1}$. Then choose $k = c^{-1/6}$.
The squared overlap value is then $(1 - c^{1/3})^3$.
So, \Cref{lemma:swaptest_rank1restrict} holds with any polynomial dominating $g(x) = 1 - (1 - (\frac{x}{1 - x})^{1/3})^3$. For example, we may use the Taylor polynomial of $g$ around $0$, or $f(x) = 4x^{1/4}$ for $x \in [0,2^{-8}]$.
\end{proof}

In fact, assuming there is a PCP to scale up the $\NP$-hardness result of~\cite{power_of_unentanglement,beigi2009np} to $\NEXP$-hardness, we have that $\QMA_{\wproper} \subseteq \EXP$, yet $\QMA_{\wproper}(2) = \NEXP$.
Just as before, we can view the power of proper states as arising from restricting $\wproper$ to pure states.

However, this class does not provide a good approach to prove $\QMA(2) = \NEXP$. Unlike $\QMA^+$ and $\QMAP$, it is not clear how to simulate $\QMA_{\wproper}$ in $\QMA$ for any completeness-soundness gap. For example, any computational basis state in $\cd(\mathbb{C}^d)$ has only $\frac{1}{d}$ squared overlap with every proper state.

\clearpage
\newpage
\section{$\QMA$ with a multiseparable quantum proof}
\label{sec:multisep}
We investigate the power of quantum proofs with another notion of multipartite (un)entanglement, called \emph{multiseparability} by \cite{Thapliyal_1999}. These are tripartite states where after tracing \emph{any} register, the reduced density matrix is separable. Label the registers $A,B,C$ as before. Then 
\begin{align*}
\cg = \{ \ket{\psi} \ |\ \Tr_A[\ketbra{\psi}] \in \sep(B,C), \Tr_B[\ketbra{\psi}] \in \sep(A,C), \Tr_C[\ketbra{\psi}] \in \sep(A,B)\}.
\end{align*}
Let $\QMAMS \defeq \QMA_{\cg}$. Then this class can also decide $\NEXP$ at some gap:
\begin{theorem}
\label{thm:qma_multisep_is_nexp}
There exist constants $1 > c > s > 0$ where $\QMAMS^{c,s} = \NEXP$.
\end{theorem}
\begin{proof}[Proof sketch]
The proof is almost identical to the proof of \Cref{thm:lowerbound}.

We consider states on \emph{five} registers $\mathbb{C}^R \otimes (\mathbb{C}^R \otimes \mathbb{C}^\kappa) \otimes (\mathbb{C}^R \otimes \mathbb{C}^\kappa)$, where $R = 2^{\poly(n)}$ and $\kappa$ is some constant. Our protocol is as follows:
\begin{itemize}
    \item With probability $p_1$, measure all the registers in the computational basis, and make sure it is of the form $\ket{v}\ket{v}\ket{c}\ket{v}\ket{c}$ (i.e.\ the vertices match and the colors match).
    \item With probability $p_2$, run $\cnot_{1,2}$ (i.e.\ applied to the first and third register), then $\cnot_{1,4}$, then $\cnot_{3,5}$, and check $\density$ on the first and third registers.
    \item Otherwise, run $\cnot_{1,2}$ and $\cnot_{1,4}$ and $\cnot_{3,5}$, and then run the constraint tests of~\cite{bassirian2023qmaplus} on the first and third registers.
\end{itemize}
Our quasirigid state before applying the $\cnot$ gates is of the form $\sum_v a_v \ket{v}\ket{v,\sigma(v)}\ket{v,\sigma(v)}$. We remark that these states are in fact \emph{multiseparable}. Since each basis element is uniquely determined by the value in any one of the registers, tracing out any one register gives a separable state. The rest of the proof goes through virtually unchanged.
\end{proof}

Remarkably, pure states that are multiseparable can be characterized in another way.
\begin{definition}
A tripartite state $\ket{\psi}$ has a \emph{generalized} Schmidt decomposition if it can be written in the form 
\begin{align*}
\ket{\psi} = \sum_i \sqrt{p_i}\ket{a_i}\ket{b_i}\ket{c_i}\,,
\end{align*}
where each set $\{\ket{a_i}\}, \{\ket{b_i}\}, \{\ket{c_i}\}$ is an orthonormal set.
\end{definition}
Note that every tripartite state can be recursively Schmidt-decomposed, but this does not imply that every set $\{\ket{a_i}\}, \{\ket{b_i}\}, \{\ket{c_i}\}$ is simultaneously an orthonormal set.
\begin{fact}[\cite{Thapliyal_1999}]
A state $\ket{\psi}$ has a generalized Schmidt decomposition iff $\ket{\psi} \in \cg$.
\end{fact}
\begin{corollary}
\label{cor:gsd_is_nphard}
There exists a constant $\alpha < 1$ above which it is $\NP$-hard to $\alpha$-approximately optimize linear functions of $\poly(n)$-dimension states with generalized Schmidt decompositions.
\end{corollary}
It is fruitful to further understand the connection between \Cref{cor:gsd_is_nphard} and other $\NP$-hardness results about testing properties of tensors.

One may also investigate the power of \emph{mixed} multiseparable states $\cmg$. 
It is likely that \Cref{lemma:swaptest_rank1restrict} can be used to show $\QMAMMS(2) = \NEXP$. The one copy-variant  $\QMAMMS$ can simulate $\QMA(2)$ by ignoring any one register of the proof, but its exact power remains unclear.

\clearpage
\newpage

\section{Computing probabilities of the $\QMAPP$ protocol for $\NEXP$}
\label{sec:probabilities}
Here, we set the probability of running each test in the protocol of \Cref{sec:qmaw_equals_nexp}. This largely follows \cite[Section 3.3]{bassirian2024superposition} although with slightly different constants. 
We rely on the following claim, paraphrased from \cite{bassirian2023qmaplus}:
\begin{lemma}[\cite{bassirian2023qmaplus}]
\label{lemma:rigid_means_qma_nexp}
 There exist constants $0 < \xi \le \cyes < 1$  and a $\QMA$ protocol $\mathsf{ConstraintCheck}$ with completeness $\cyes$ and soundness $\cyes-\xi$ that decides $\NEXP$, \emph{assuming} the quantum proof is a \emph{rigid} state.
\end{lemma}
\Cref{lemma:rigid_means_qma_nexp} implies the existence of absolute constants $\cyes, \xi,\kappa$. We choose \emph{distance thresholds} $\nul, \nuh$ to be small positive constants that satisfy the following conditions:
\begin{align*}
0 \leq \nuh \le \frac{\nul}{\kappa}\,, & &  \frac{\nuh}{\nul} \le \frac{\xi}{6(1-\cyes)}\,, & & \left( \kappa \nul + (\kappa + 1) \sqrt{(\kappa + 1)\nul}\right)^{1/2} \le \frac{\xi}{2}\,. \end{align*}
This can be done by first making both constants equal and small enough to satisfy the last inequality, then reducing $\nuh$ to satisfy the first two inequalities.

Now we choose the probabilities $p_1, p_2, p_3$. Let
\begin{align*}
    p_1= \frac{1}{Z} & & p_2= \frac{(\nul + \nuh) }{\nuh^2  Z} & & p_3 = \frac{\nul}{2(1 - \cyes)Z}\,,
\end{align*}
where $Z \defeq 1 + \frac{(\nul + \nuh)}{\nuh^2}  + \frac{\nul}{2(1 - \cyes)}$, so that the probabilities sum to $1$.

Given a quantum proof $\ket{\psi}$, the verifier accepts with probability 
\begin{align*}
    p_1 \cdot \Pr[\density \text{ succeeds}] + p_2 \Pr[\semicheck \text{ succeeds}] + p_3 \cdot \Pr[\mathsf{ConstraintCheck} \text{ succeeds}]\,.
\end{align*}
In completeness, we consider the quantum proof guaranteed by \Cref{lemma:rigid_means_qma_nexp}. Since this state is rigid, the verifier accepts with probability at least $\pyes \defeq p_1 \cdot \frac{1}{\kappa} + p_2 \cdot 1 + p_3 \cdot \cyes$. For soundness, we analyze four separate cases:
\begin{enumerate}
    \item In the first case, the proof doesn't test well enough on $\density$. If $\Pr[\density \text{ succeeds}] = \frac{1}{\kappa} - d$ for $d \ge \nul$, then the verifier accepts with probability at most
    \begin{align*}
        \pno \le p_1 \cdot (\frac{1}{\kappa} - d) + p_2 \cdot 1 + p_3 \le \pyes - p_1 \cdot \nul + p_3 \cdot (1 - \cyes) = \pyes - \frac{\nul}{2 Z}\,.
    \end{align*}
    \item In the second case, the proof tests too well on $\density$, and so cannot test well enough on $\semicheck$. If $\Pr[\density \text{ succeeds}] =  \frac{1}{\kappa} + d$ for $d \ge \nuh$, then by \Cref{lemma:quadratic} we can upper bound the success probability of $\semicheck$ as $\Pr[\semicheck \text{ succeeds}] \le \frac{k+1-d^2}{k+1}$. Then
    \begin{align*}
        \pno \le p_1 \cdot \left( \frac{1}{\kappa} + d\right) + p_2 \cdot \frac{\kappa+1-d^2}{\kappa+1} + p_3 \le \pyes + p_1 \cdot d - p_2 \cdot \frac{d^2}{\kappa+1} + p_3 \cdot (1 - \cyes) \,.
    \end{align*}
The right-most expression depends on $d$ through the term $p_1 \cdot d - p_2 \cdot \frac{d^2}{\kappa+1} = p_1 \left( d - \frac{(\nul + \nuh) d^2}{(\kappa+1)\nuh^2}  \right)$. 
This term is decreasing with $d$ when $d > \frac{\kappa+1}{2}\cdot\frac{\nuh}{1 + \nul/\nuh} \geq \frac{\nuh}{2}$. So the expression is maximized at $d = \nuh$, and then
\begin{align*}
     \pno \le  \pyes + p_1 \left( \nuh - \frac{(\nul + \nuh)\nuh^2}{\nuh^2} \right)+ p_3 \cdot (1 - \cyes) =  \pyes - \frac{\nul}{2 Z}\,.
\end{align*}
\item In the third case, the proof doesn't test well enough on $\semicheck$, and the acceptance probability of $\density$ is not too high. If $\Pr[\density \text{ succeeds}] = \frac{1}{\kappa} + d$ for $-\nul \le d \le \nuh$ and $\Pr[\semicheck \text{ succeeds}] \le 1-\nul$, then
\begin{align*}
    \pno &\le p_1 \cdot \left( \frac{1}{\kappa} + d \right) + p_2 \cdot \left(1- \nul\right) + p_3 \le   \pyes + p_1 \cdot d - p_2 \cdot \nul +  p_3 \cdot (1 - \cyes)\,.
\end{align*}
By our choice of constants, $d \le \nuh \le \nul$. Moreover, one can verify that $p_2 \ge 2 \cdot p_1$ whenever $\nuh \le 1$. So $p_1 \cdot d - p_2 \cdot \nul \le - p_1\cdot  \nul$, and $\pno \le \pyes - p_1 \cdot \nul +  p_3 \cdot (1 - \cyes) =   \pyes - \frac{\nul}{2 Z}$.

\item In the fourth case, the proof tests well on $\semicheck$, and the acceptance probability of $\density$ is not too low. So the proof must be close to \emph{rigid}, and do worse on $\mathsf{ConstraintCheck}$. If $\Pr[\density \text{ succeeds}] = \frac{1}{\kappa} + d$ for $-\nul \le d \le \nuh$ and $\Pr[\semicheck \text{ succeeds}] \ge 1-\nul$, then by \Cref{claim:rigidity}, there is a rigid state $\ket{\chi}$ such that 
\begin{align*}
    |\braket{\chi}{\psi}|^2 \ge 1 - \kappa d - (\kappa + 1) \sqrt{(\kappa + 1)\nul}\,.
\end{align*}
We use \Cref{fact:fuchs_vdg_implication} with $\Pi$ equal to the accepting projector of $\mathsf{ConstraintCheck}$; then we have $|\Pr[\mathsf{ConstraintCheck} \text{ succeeds}] - \Pr[\mathsf{ConstraintCheck} \text{ succeeds}]| \le \left(\kappa d + (\kappa + 1) \sqrt{(\kappa + 1)\nul}\right)^{1/2}$. Putting this all together,
\begin{align*}
     \pno &\le p_1 \cdot \left( \frac{1}{\kappa} + d \right) + p_2 \cdot 1 + p_3 \cdot\left( \cyes - \xi + \left(\kappa d + (\kappa + 1) \sqrt{(\kappa + 1)\nul}\right)^{1/2}\right)
     \\
     &\le \pyes + p_1 \cdot d + p_3 \cdot \left( - \xi + \left(\kappa d + (\kappa + 1) \sqrt{(\kappa + 1)\nul}\right)^{1/2}\right) 
     \\
     &\le \pyes + p_1 \cdot \nuh - p_3 \cdot \frac{\xi}{2} 
     \\
     &= \pyes + \frac{1}{Z} \left( \nuh - \frac{\nul \cdot \xi}{4(1 - \cyes)} \right)
     \\
     &\le \pyes - \frac{\nuh}{2Z} \,.
\end{align*}
\end{enumerate}
So, this protocol succeeds with completeness $\pyes$ and soundness at most $\pyes - \frac{\nuh}{2Z}$.

\clearpage
\newpage

\section{Padding argument for the proof of $\QMAP(2) = \NEXP$}
\label{sec:padding}
Here, we verify that the protocol in \Cref{sec:qmaw_equals_nexp} goes through even when the quantum proofs are from $\purerestrict(\cm^\star)$; i.e. states in $\cpp$ obeying the register size inequality in the proof of \Cref{cor:qmapurif2_is_nexp}.
We modify the protocol as follows. We prepend a scratch qudit (using index variable $z$) to the first register to satisfy the register size inequality.
In both $\density$ and $\semicheck$, we also measure the scratch qudit in the standard basis, and reject if it does not output $0_z$.
In the other test, we ignore the scratch qudit.

For this modified protocol, we prove that \Cref{lemma:semicheck_means_close_to_quasirigid}, \Cref{claim:rigidity}, and \Cref{lemma:quadratic} hold with the same parameters.
We may then choose the test probabilities exactly as in \Cref{sec:probabilities}, as the analysis only depends on these three statements.

\begin{lemma}[Padded version of \Cref{lemma:semicheck_means_close_to_quasirigid}] \label{lemma:padded_semicheck_means_close_to_quasirigid}
    Suppose $\semicheck$ succeeds with probability $1-\epsilon$, and $\rho_{B, C}$ is separable. Then there is a state $\ket{0_z}\ket{\phi}$ that is \emph{quasirigid} after applying $\cnot_{1,3}\cnot_{2,4}$ such that $|\bra{0_z,\phi}\ket{\psi}|^2 \ge 1 - (\kappa + 1) \epsilon$.
\end{lemma}
\begin{proof}
Label the computational basis elements of $\ket{\psi}$ as $a_{zvc,wd}\ket{z}\ket{v}\ket{c,w}\ket{d}$.
Let $\sigma: [R] \to [\kappa]$ be a function that picks out the largest amplitude per vertex when the scratch qudit $z = 0$, i.e. where $|a_{0_z,v\sigma(v), v\sigma(v)}|^2 \ge |a_{0_z,vd, vd}|^2$ for all $d \ne \sigma(v)$. For convenience, denote $\sigma_v \defeq \sigma(v)$.

Consider the matrix elements of the reduced density matrix of $\ket{\psi}$ after tracing out the first register:
\begin{align*}
        \bra{c, v, d}\rho_{B, C}\ket{e, w, f} = \sum_z \sum_x a_{z,xc,vd} a_{z,xe,wf}^\dagger\,.
\end{align*}
We use this expression to upper-bound the weight on colors not chosen by $\sigma$:
\begin{align*}
\sum_v \sum_{d;d \ne \sigma_v}    |a_{0_z,vd, vd}|^2 
\le
\sum_v \sum_{d;d \ne \sigma_v}    |a_{0_z,v\sigma_v, v\sigma_v}| \cdot |a_{0_z,vd, vd}^\dagger |
=
\sum_v \sum_{d;d \ne \sigma_v}    |a_{0_z,v\sigma_v, v\sigma_v} \cdot a_{0_z,vd, vd}^\dagger |\,.
\end{align*}
We break this upper bound into three parts, by rewriting $|a_{0_z,v\sigma_v, v\sigma_v} \cdot a_{0_z,vd, vd}^\dagger |$ as
\begin{align*}
        &|\bra{\sigma_v, v, \sigma_v}\rho_{B, C}\ket{d, v, d} 
        -\sum_{z \ne 0}\sum_x a_{z,x\sigma_v,v\sigma_v} \cdot a_{z,xd,vd}
        - \sum_{x;x \ne v} a_{0_z,x\sigma_v,v\sigma_v} \cdot a_{0_z,xd,vd}^\dagger|
\\
\le\ 
&|\bra{\sigma_v, v, \sigma_v}\rho_{B, C}\ket{d, v, d}| 
+ \sum_{z \ne 0}\sum_x |a_{z,x\sigma_v,v\sigma_v}| \cdot |a_{z,xd,vd}^\dagger|
+ \sum_{x;x \ne v} |a_{0_z,x\sigma_v,v\sigma_v}| \cdot |a_{0_z,xd,vd}^\dagger|
\,.
\end{align*}
\begin{enumerate}
    \item For the first term, since $\rho_{B, C}$ is separable, we use one of the bounds proven in \Cref{lemma:separable_offdiagonal_cauchy_schwarz}:
\begin{align*}
        |\bra{\sigma_v, v, \sigma_v}\rho_{B, C} \ket{d, v, d}| \le \frac{1}{2} \cdot \big(\bra{\sigma_v, v, d}\rho_{B, C}\ket{\sigma_v, v, d} + \bra{d, v, \sigma_v}\rho_{B, C}\ket{d, v, \sigma_v}\big)\,.
\end{align*}
So then
\begin{align*}
    \sum_v \sum_{d;d \ne \sigma_v}  |\bra{\sigma_v, v, \sigma_v}\rho_{B, C} \ket{d, v, d}
    &\le
    \frac{1}{2} \left(\sum_v \sum_{d;d \ne \sigma_v} 
     \bra{\sigma_v, v, d}\rho_{B, C}\ket{\sigma_v, v, d} + \bra{d, v, \sigma_v}\rho_{B, C}\ket{d, v, \sigma_v}\right)
    \\
    &\le 
     \frac{1}{2} \left(\sum_v \sum_{d;d \ne \sigma_v} 
     \sum_z \sum_{x} |a_{z,x\sigma_v,vd}|^2 
     + |a_{z,xd,v\sigma_v}|^2\right)\,.
\end{align*}
\item 
The second term appears because of the scratch qudit. This term is small because $z \ne 0$:
\begin{align*}
    \sum_v \sum_{d; d \ne \sigma_v} \sum_{z \ne 0} \sum_x |a_{z,x \sigma_v,v \sigma_v}| \cdot |a_{z,x d,v d}^\dagger| 
    \le \sum_{z \ne 0} \sum_{v,x,d} |a_{z,x d,v d}|^2
\end{align*}
The sum of the first and second term are upper-bounded by $\sum_{(z,x,c) \ne  (0,v,d)} |a_{z,xc,vd}|^2 = \epsilon$.
\item 
For the third term, we use the AM-GM inequality:
\begin{align*}
    |a_{0_z,x\sigma_v,v\sigma_v}| \cdot |a_{0_z,xd,vd}^\dagger| \le \frac{1}{2} \cdot \left(|a_{0_z,x\sigma_v,v\sigma_v}|^2 +  |a_{0_z,xd,vd}^\dagger|^2 \right)\,.
\end{align*}
Taking $|a_{0_z,xd,vd}|^2$ over all three sums $(v,d,x)$ is at most $\sum_{(z,x,c) \ne  (0,v,d)} |a_{z,xc,vd}|^2 = 1 - (1-\epsilon) = \epsilon$. 
Similarly, summing $|a_{0_z,x\sigma_v,v\sigma_v}|^2$ over sums $(v,x)$ contributes at most $\epsilon$, so at most $(\kappa-1) \epsilon$ over sums $(v,d \ne \sigma_v,x)$.
In total, this term is at most $\frac{\kappa}{2}\epsilon$.
\end{enumerate}
Combining each of the three terms, we have
\begin{align*}
    \sum_v \sum_{d;d \ne \sigma_v}    |a_{0_z,vd, vd}|^2 \le \epsilon + \frac{\kappa}{2}\epsilon = \frac{\kappa + 2}{2}\epsilon \le \kappa \epsilon\,,
\end{align*}
where the last inequality holds because $\kappa \ge 2$ (i.e. we need at least two colors).

Consider the state $\ket{\phi} = \frac{1}{\sqrt{\gamma}} \sum_v a_{0_z,v \sigma_v, v\sigma_v} \ket{v}\ket{\sigma_v}\ket{v}\ket{\sigma_v}$, where $\gamma = \sum_v |a_{0_z,v \sigma_v, v\sigma_v}|^2$ is chosen so that $\braket{\phi}{\phi} = 1$. Note that $\ket{\phi}$ is quasirigid after applying $\cnot_{1,3}$ and $\cnot_{2,4}$. Moreover, $|\braket{0_z,\phi}{\psi}|^2 = \gamma$. Putting everything together,
\begin{align*}
    \gamma = \sum_v |a_{0_z,v \sigma_v, v\sigma_v}|^2 = \Pr[\semicheck \text{ succeeds}] -  \sum_v \sum_{d;d \ne \sigma_v}    |a_{0_z,vd, vd}|^2  \ge 1 - \epsilon - \kappa \epsilon\,.\tag*{\qedhere} 
\end{align*}
\end{proof}

\begin{claim}[Padded version of \Cref{claim:rigidity}]
\label{claim:rigidity_padded}
Suppose $\density$ and $\semicheck$ succeed on $\ket{\psi}$ with probability $\frac{1}{\kappa} - d_{\mathsf{D}}$ and $1- d_\mathsf{M}$, respectively. Then there is a state $\ket{0_z}\ket{\chi}$ that is rigid after applying $\cnot_{1,3}\cnot_{2,4}$ such that 
    \begin{align*}
    |\braket{0_z,\chi}{\psi}|^2 &\ge 1 - \kappa d_{\mathsf{D}}  - (\kappa+1) \sqrt{(\kappa + 1) d_\mathsf{M}}\,.
    \end{align*}
\end{claim}
\begin{proof}
    By \Cref{lemma:padded_semicheck_means_close_to_quasirigid}, there exists a state $\ket{\phi}=\sum_{z,v}\alpha_{z,v}\ket{z,v}\ket{c_{z,v},v}\ket{c_{z,v}}$ such that 
    \begin{align*}
        |\braket{\psi}{\phi}|^2 \ge 1 - (\kappa+1)d_\mathsf{M}\,.
    \end{align*}
    By \Cref{fact:fuchs_vdg_implication}, 
    for any  unit vector $\ket{\mu}$,
    we have
    $\left| |\braket{\mu}{\psi}|^2 - |\braket{\mu}{\phi}|^2 \right| \le \sqrt{(\kappa+1)d_\mathsf{M}}$.
    We use this in two places. 
    Let $\ket{+'} = \ket{0_z} \otimes \left(\cnot_{1,3}\cnot_{2,4}\ket{+}\ket{0}\ket{0}\right)$. First, since  $|\braket{+'}{\psi}|^2 \ge \frac{1}{\kappa} - d_{\mathsf{D}}$, applying \Cref{fact:fuchs_vdg_implication} with $\ket{\mu} = \ket{+'}$ implies $|\braket{+'}{\phi}|^2 \ge \frac{1}{\kappa} - d_{\mathsf{D}} -\sqrt{(\kappa+1)d_\mathsf{M}}$. 
    Now let $\ket{\chi}=\frac{1}{\sqrt{R}}\sum_v \ket{v}\ket{c_{0_z,v},v}\ket{c_{0_z,v}}$ be such that $\ket{0_z}\ket{\chi}$ has exactly the shared computational basis elements of $\ket{+'}$ and $\ket{\phi}$. Then
    \begin{align*}
        |\braket{0_z,\chi}{\phi}|^2 =\left|\frac{1}{\sqrt{R}}\sum_{v\in[R]}\alpha_{0_z,v} \right|^2
        =\kappa \left|\frac{1}{\sqrt{\kappa R}} \sum_{v\in[R]}^R\alpha_{0_z,v}\right|^2
        = \kappa |\braket{+'}{\phi}|^2
        \ge 1 - \kappa\left(d_{\mathsf{D}} +\sqrt{(\kappa+1)d_\mathsf{M}}\right)\,.
    \end{align*}
   Applying \Cref{fact:fuchs_vdg_implication} using $\ket{\mu} = \ket{0_z}\ket{\chi}$ then implies
   \begin{align*}
       |\braket{0_z,\chi}{\psi}|^2  \ge  1 - \kappa\left(d_{\mathsf{D}} +\sqrt{(\kappa+1)d_\mathsf{M}}\right) - \sqrt{(\kappa+1)d_\mathsf{M}}=1 - \kappa d_{\mathsf{D}}  - (\kappa+1) \sqrt{(\kappa + 1) d_\mathsf{M}}\,.\tag*{\qedhere}
   \end{align*}
\end{proof}

\begin{lemma}[Padded version of \Cref{lemma:quadratic}]
\label{lemma:quadratic_padded}
    Suppose $\density$ and $\semicheck$ succeed on state $\ket{\psi}$ with probability $w_{\mathsf{D}} \ge \frac{1}{\kappa}$ and $w_{\mathsf{M}}$, respectively. Then $(w_{\mathsf{D}} -  \frac{1}{\kappa})^2 +  (\kappa+1)w_{\mathsf{M}}\le  \kappa+1$.
\end{lemma}
\begin{proof}
By \Cref{lemma:padded_semicheck_means_close_to_quasirigid}, there exists a $\ket{\phi} = \sum_{z,v} \alpha_{z,v} \ket{z,v}\ket{c_{z,v},v}\ket{c_{z,v}}$ such that 
\begin{align*}
    |\braket{\psi}{\phi}|^2 \ge 1-(\kappa+1)(1 - w_{\mathsf{M}})\,.
\end{align*}
Using \Cref{fact:fuchs_vdg_implication}, we see that 
$\left| |\braket{+'}{\psi}|^2 - |\braket{+'}{\phi}|^2 \right|$ is at most $\sqrt{(\kappa+1)(1-w_{\mathsf{M}})}$. Note that by assumption, $|\braket{+'}{\psi}|^2 = w_{\mathsf{D}}\ge\frac{1}{\kappa}$, and by Cauchy-Schwarz, $|\braket{+'}{\phi}|^2 = \frac{1}{R \cdot \kappa} \left| \sum_{v \in [R]} \alpha_{0_z,v} \right|^2 \le \frac{1}{\kappa} \sum_{v \in [R]} |\alpha_{0_z,v}|^2 \le \frac{1}{\kappa}$. Combining these claims, we have 
\begin{align*}
    w_{\mathsf{D}} - \frac{1}{\kappa} \le |\braket{+'}{\psi}|^2 - |\braket{+'}{\phi}|^2 \le  \sqrt{(\kappa+1)(1-w_{\mathsf{M}})}\,.
\end{align*}
The lemma follows after rearranging the terms.
\end{proof}
\end{document}